\documentclass[11pt]{article}

\usepackage{amsmath,amssymb,amsthm,amsfonts,latexsym,bbm,xspace,graphicx,float,mathtools,braket,dirtytalk}
\usepackage[inkscapeformat=png]{svg}
\usepackage[letterpaper,margin=1in]{geometry}

\usepackage{enumitem}

\usepackage{yhmath}
\usepackage{thmtools}
\usepackage{thm-restate}
\usepackage{subcaption}
\usepackage[linesnumbered,ruled,vlined]{algorithm2e}
\SetKwInput{KwInput}{Input}                %
\SetKwInput{KwOutput}{Output}
\SetKwFor{RepTimes}{repeat}{times}{end}

\usepackage{svg}

\usepackage[colorinlistoftodos]{todonotes}

\usepackage[pagebackref,colorlinks,citecolor=blue,,linkcolor=magenta,bookmarks=true]{hyperref}
\usepackage[nameinlink,capitalize]{cleveref}

\renewcommand{\backref}[1]{}

\renewcommand{\backrefalt}[4]{%
\ifcase #1 %
\or
[p.\ #2]%
\else
[pp.\ #2]%
\fi}

\newtheorem{theorem}{Theorem}[section]
\newtheorem{lemma}[theorem]{Lemma}

\newtheorem{proposition}[theorem]{Proposition}
\newtheorem{fact}[theorem]{Fact}
\newtheorem{corollary}[theorem]{Corollary}

\theoremstyle{definition}
\newtheorem{definition}[theorem]{Definition}

\renewcommand{\Pr}{\mathop{\bf Pr\/}}
\newcommand{\E}{\mathop{\bf E\/}}
\newcommand{\Ex}{\mathop{\bf E\/}}

\newcommand{\tr}{\mathrm{tr}} \newcommand{\Tr}{\tr} 
\newcommand{\poly}{\mathrm{poly}}
\newcommand{\negl}{\mathrm{negl}}

\newcommand{\R}{\mathbb R}
\newcommand{\C}{\mathbb C}

\newcommand{\Z}{\mathbb Z}
\newcommand{\F}{\mathbb F}

\newcommand{\NP}{\mathsf{NP}}

\newcommand{\eps}{\varepsilon}

\newcommand{\wh}[1]{\widehat{#1}}

\newcommand{\calC}{\mathcal{C}}
\newcommand{\calD}{\mathcal{D}}

\newcommand{\calS}{\mathcal{S}}

\newcommand{\stabset}{\mathcal{S}}

\newcommand{\weyl}{\mathrm{Weyl}}

\newcommand{\fidelity}{F}

\newcommand{\tracedistance}[1]{d_{\mathrm{tr}}(#1)}

\newcommand{\indic}[1]{\mathbbm{1}_{#1}}   %

\newcommand{\abs}[1]{\lvert #1 \rvert}

\newcommand{\ketbra}[2]{\ket{#1}\!\!\bra{#2}}

\renewcommand{\hat}{\widehat}

\newcommand{\sympcomp}{\perp}

\newcommand{\ignore}[1]{}

\newcommand{\anote}[1]{}
\newcommand{\jnote}[1]{}
\newcommand{\hnote}[1]{}

\newcommand{\ainnote}[1]{}
\newcommand{\jinnote}[1]{}
\newcommand{\hinnote}[1]{}
\newcommand{\tnote}[1]{}
\newcommand{\znote}[1]{}

\newcounter{termcounter}[equation]
\renewcommand{\thetermcounter}{\the\numexpr\value{equation}+1\relax.\roman{termcounter}}
\crefname{term}{term}{terms}
\creflabelformat{term}{#2\textup{(#1)}#3}

\makeatletter
\def\term{\@ifnextchar[\term@optarg\term@noarg}%
\def\term@optarg[#1]#2{%
  \textup{#1}%
  \def\@currentlabel{#1}%
  \def\cref@currentlabel{[][2147483647][]#1}%
  \cref@label[term]{#2}}
\def\term@noarg#1{%
  \refstepcounter{termcounter}%
  \textup{\thetermcounter}%
  \cref@label[term]{#1}}
\makeatother

\title{
Improved Stabilizer Estimation via Bell Difference Sampling
}

\author{Sabee Grewal\thanks{University of Texas at Austin. \texttt{sabee@cs.utexas.edu}.}\and Vishnu Iyer\thanks{University of Texas at Austin. \texttt{vishnu.iyer@utexas.edu}.}\and William Kretschmer\thanks{University of Texas at Austin, Simons Institute for the Theory of Computing, and University of California, Berkeley. \texttt{kretsch@berkeley.edu}.}\and Daniel Liang\thanks{University of Texas at Austin and Rice University. \texttt{dliang@utexas.edu}}.}

\date{}

\begin{document}

\maketitle

\begin{abstract}
We study the complexity of learning quantum states in various models with respect to the stabilizer formalism and obtain the following results:
\begin{itemize}
    \item We prove that a linear number of $T$-gates are necessary for any Clifford+$T$ circuit to prepare computationally pseudorandom quantum states, an exponential improvement over the previously known bound. This bound is asymptotically tight if linear-time quantum-secure pseudorandom functions exist.
    \item Given an $n$-qubit pure quantum state $\ket\psi$ that has fidelity $\tau$ with some stabilizer state, we give an algorithm that outputs a succinct description of a stabilizer state that witnesses fidelity at least $\tau - \eps$. The algorithm uses $O(n/(\eps^2\tau^4))$ samples and $\exp\left(O(n/\tau^4)\right) / \eps^2$ time. 
    In the regime of $\tau$ constant, this algorithm estimates stabilizer fidelity substantially faster than the na\"ive $\exp(O(n^2))$-time brute-force algorithm over all stabilizer states.
    \item In the special case of $\tau > \cos^2(\pi/8)$, we show that a modification of the above algorithm runs in polynomial time.
    \item 
    We exhibit a tolerant property testing algorithm for stabilizer states.
\end{itemize}

The underlying algorithmic primitive in all of our results is Bell difference sampling. 
To prove our results, 
we establish and/or strengthen connections between Bell difference sampling, symplectic Fourier analysis, and graph theory. 
\end{abstract}

 \newpage

\section{Introduction}
A central goal in quantum information is to understand which quantum states are efficiently learnable. While many quantum state learning algorithms are extremely efficient in sample complexity \cite{aaronson2018shadow, buadescu2021improved, HKP20-classical-shadows}, fewer classes of time-efficiently-learnable quantum states are known. One such example is the class of \textit{stabilizer states}, which are $n$-qubit states that are stabilized by a group of $2^n$ commuting Pauli matrices.\footnote{Some other examples of state classes that admit time-efficient learning algorithms include matrix product states \cite{cramer2010efficient}, non-interacting fermion states \cite{aaronson2023efficient}, and certain classes of phase states \cite{arunachalam2022phase}.} 
Stabilizer states are well-studied because of their broad importance and widespread applications throughout quantum information, including in quantum error correction \cite{shor1995codes, calderbank1996codes, Got97-thesis}, efficient classical simulation of quantum circuits \cite{bravyi2016trading, Bravyi2019simulationofquantum}, randomized benchmarking \cite{knill2008benchmarking}, and measurement-based quantum computation \cite{raussendorf2000mbqc}, to name a few examples.

The first computationally efficient algorithm for learning a complete description of an unknown stabilizer state was given by Montanaro \cite{montanaro-bell-sampling}.\footnote{In 2008, Gottesman gave a short video lecture explaining how to learn stabilizer states, based on joint work with Aaronson \cite{aaronson43identifying}. However, the details of this algorithm were never published.} Given copies of a stabilizer state $\ket{\phi}$, Montanaro's algorithm utilizes the algebraic properties of Pauli matrices and Bell-basis measurements to efficiently learn the generators of the stabilizer group of $\ket{\phi}$, which suffices to determine $\ket{\phi}$.
More specifically, Montanaro (implicitly) introduced \emph{Bell difference sampling}, which, at a high level, is an algorithmic primitive that takes copies of some state and induces a measurement distribution on Pauli matrices.
Bell difference sampling was studied more thoroughly in \cite{gross2021schur} and has seen extended success in the development of algorithms for stabilizer states and states that are close to stabilizer states \cite{montanaro-bell-sampling, gross2021schur, lai2022learning,grewal_et_al:LIPIcs.ITCS.2023.64, haug2023scalable}.

In this work, we extend the use of Bell difference sampling to give faster, more general, and otherwise improved algorithms for learning properties of quantum states related to the stabilizer formalism.
By understanding how these properties affect the Bell difference sampling distribution, we are able to find relevant certificates of these properties faster than the previous state-of-the-art.

\subsection{Our Results}\label{subsec:our-results}

\subsubsection*{Tight Pseudorandomness Bounds}

Pseudorandom states are a quantum cryptographic primitive that have recently attracted much attention in quantum cryptography and complexity theory. They can be thought of as a quantum analogue of pseudorandom generators, with the main difference being that pseudorandom states mimic the Haar measure over $n$-qubit states, rather than the uniform distribution over $n$-bit strings. 
Formally, they are defined as follows:
\begin{definition}[Pseudorandom quantum states \cite{ji-pseudorandom-states2018}]
\label{def:prs}
A keyed family of $n$-qubit quantum states $\{\ket{\varphi_k}\}_{k \in \{0,1\}^\kappa}$ is \emph{pseudorandom} if the following conditions hold:
\begin{enumerate}
\item (Efficient generation) There is a polynomial-time quantum algorithm $G$ that generates $\ket{\varphi_k}$ on input $k$ (so, in particular, $n \le \poly(\kappa)$).
\item (Computational indistinguishability) For any $\poly(\kappa)$-time quantum adversary $\mathcal{A}$ and $T = \poly(\kappa)$:
\[\left| \Pr_{k \sim \{0,1\}^\kappa}\left[\mathcal{A}\left(\ket{\varphi_k}^{\otimes T}\right) = 1 \right] - \Pr_{\ket{\psi} \sim \mu_{\mathrm{Haar}}^n}\left[\mathcal{A}\left(\ket{\psi}^{\otimes T}\right) = 1 \right] \right| = \negl(\kappa),\]
where $\mu_{\mathrm{Haar}}^n$ denotes the $n$-qubit Haar measure, and $\negl(\kappa)$ denotes an arbitrary negligible function of $\kappa$.
\end{enumerate}
\end{definition}
Pseudorandom states suffice to build a wide range of cryptographic primitives, including quantum commitments, secure multiparty computation, one-time digital signatures, and more \cite{ji-pseudorandom-states2018,ananth2021crypto,morimae2021commitments,bartusek21oneway,grilo21oblivious,HMY22-qpke}. The language of pseudorandom states has also been found to play a key role in resolving some paradoxes at the heart of black hole physics \cite{BFV20-prs-wormhole,Bra22-decoding}. 
Finally, and perhaps most surprisingly, there is recent evidence to suggest that pseudorandom states can be constructed without assuming the existence of one-way functions \cite{Kre21-pseudorandom,KQST23-prs}. 

Collectively, these results have motivated recent works that seek to characterize what computational properties or resources are required of pseudorandom states. For example, \cite{aaronson2022quantumpseudo} investigates the possibility of building pseudorandom quantum states with limited entanglement, and prove the existence of pseudorandom state ensembles with entanglement entropy substantially smaller than $n$, assuming the existence of quantum-secure one-way functions.

Analogously, Grewal, Iyer, Kretschmer, and Liang \cite{grewal_et_al:LIPIcs.ITCS.2023.64} study quantum pseudorandomness from the perspective of stabilizer complexity. They treat the number of non-Clifford gates in a circuit as a resource, similar to size or depth. The main result of \cite{grewal_et_al:LIPIcs.ITCS.2023.64} shows that states having fidelity at least $\frac{1}{\poly(n)}$ with a stabilizer state cannot be computationally pseudorandom. As a consequence, they deduce that $\omega(\log n)$ non-Clifford gates are necessary for a family of circuits to yield an ensemble of pseudorandom quantum states.

We give an exponential improvement on this lower bound:\footnote{We remark that while the result of \cite{grewal_et_al:LIPIcs.ITCS.2023.64} is not tight in terms of the number of non-Clifford gates, recent work \cite{aaronson2022quantumpseudo} shows that \cite{grewal_et_al:LIPIcs.ITCS.2023.64}'s bound in terms of stabilizer fidelity is optimal up to polynomial factors, because \cite{aaronson2022quantumpseudo} constructs pseudorandom state ensembles with any inverse-superpolynomial stabilizer fidelity (assuming quantum-secure one-way functions exist).}

\begin{theorem}[Informal version of \cref{cor:prs-linear-lower-bound}]
\label{thm:pseudorandomness-lb-intro}
Any family of Clifford circuits that produces an ensemble of $n$-qubit computationally pseudorandom quantum states must use at least $n/2$ auxiliary non-Clifford single-qubit gates.
\end{theorem}

In the special case that the non-Clifford gates are all diagonal (e.g.\ $T$-gates), our lower bound improves to $n$.

Under plausible computational assumptions, \Cref{thm:pseudorandomness-lb-intro} is tight up to constant factors. In particular, the existence of linear-time quantum-secure pseudorandom functions implies the existence of linear-time constructible pseudorandom states \cite{brakerski10.1007/978-3-030-36030-6_10,grewal_et_al:LIPIcs.ITCS.2023.64}, which of course have at most $O(n)$ non-Clifford gates. Note that linear-time \textit{classically}-secure pseudorandom functions are strongly believed to exist \cite{ishai_10.1145/1374376.1374438, fan_10.1145/3519935.3520010}, and it seems conceivable that these constructions remain secure against quantum adversaries.

We remark that \Cref{thm:pseudorandomness-lb-intro} bears analogy to a recent result of Leone, Oliviero, Lloyd, and Hamma \cite{leone-stabilizer-nullity} that information scrambled by an $n$-qubit unitary implemented with Clifford gates and $t < n$ $T$-gates can be efficiently unscrambled. In particular, both \Cref{thm:pseudorandomness-lb-intro} and \cite{leone-stabilizer-nullity} establish different forms of non-pseudorandomness (for states and unitaries, respectively) in the same parameter regime of non-Cliffordness.

\subsubsection*{Faster Stabilizer State Approximation}

As noted earlier, one of the prominent applications of stabilizer states is in classical simulation algorithms of quantum circuits. Such algorithms work by modeling the output state of a quantum circuit as a decomposition of stabilizer states (e.g., as a linear combination) \cite{Bravyi2019simulationofquantum}. The runtime of these algorithms then scale with respect to one of several measures of the ``amount of non-stabilizerness'' in this decomposition. These measures are sometimes called \textit{magic monotones} \cite[Definition 3]{Veitch_2014} \cite[Definition 3]{gu2023little}, because they are non-increasing under Clifford operations. Typically, magic monotones increase exponentially as non-Clifford gates are applied.\footnote{Some authors prefer to work with the logarithm of the monotone, so that they scale linearly as non-Clifford gates are applied.} Examples of well-known magic monotones include the stabilizer rank, stabilizer extent, and inverse of stabilizer fidelity \cite{Bravyi2019simulationofquantum}.

A series of recent and simultaneous works have explored the question of whether magic monotones can be estimated efficiently, or whether states with low magic are efficiently learnable. 
For example, recall that \cite{grewal_et_al:LIPIcs.ITCS.2023.64} showed that states with non-negligible stabilizer fidelity are weakly learnable, in the sense that they are efficiently distinguishable from random.
\cite{grewal2023efficient,grewal2023efficient2,hangleiter2023bell,leone2023learning} proved that states with bounded stabilizer nullity are efficiently learnable, and \cite{grewal2023efficient} also gave an efficient property tester for stabilizer nullity. \cite{gu2023little} showed that various magic monotones \textit{cannot} be estimated efficiently in certain parameter regimes, by constructing states with low magic that are cryptographically indistinguishable from states with large magic. Finally, \cite{arunachalam2022phase,anshu2023survey} raised the question of whether states of bounded stabilizer rank are efficiently learnable.

Our second result is a further contribution towards understanding the learnability of low-magic states: we give an algorithm that finds stabilizer state approximations of states with non-negligible \textit{stabilizer fidelity}.
As its name suggests, stabilizer fidelity (denoted $F_\stabset(\ket{\psi})$) measures how close a state $\ket{\psi}$ is to a stabilizer state: it is simply the maximum of $\abs{\braket{\phi|\psi}}^2$ over all stabilizer states $\ket{\phi}$. Hence, it is not hard to see that the inverse of stabilizer fidelity is a magic monotone. Assuming $\ket{\psi}$ has stabilizer fidelity at least $\tau$, our algorithm returns a stabilizer state that witnesses overlap at least $F_\stabset(\ket{\psi}) - \eps$ with $\ket{\psi}$.

\begin{theorem}[Informal version of \cref{thm:fidelity_learning}]
Fix $\tau > \eps > 0$. There is an algorithm that, given copies of an $n$-qubit pure state $\ket{\psi}$ with $F_\stabset(\ket{\psi}) \ge \tau$, returns a stabilizer state $\ket\phi$ that satisfies $\abs{\braket{\phi|\psi}}^2 \geq F_\stabset(\ket{\psi}) - \eps$ with high probability. The algorithm uses $O(n/(\eps^2\tau^4))$ copies of $\ket\psi$ and $\exp\left(O(n/\tau^4)\right) / \eps^2$ time.
\label{thm:fidelity_learning_informal}
\end{theorem}

To our knowledge, this is the first nontrivial algorithm to approximate an arbitrary quantum state with a stabilizer state.\footnote{We thank David Gosset (personal communication) for bringing this barrier to our attention.} Indeed, we are not aware of any prior algorithm better than 
a brute-force search over all stabilizer states, which takes $2^{O(n^2)}$ time and $O(n^2)$ samples.\footnote{The polynomial sample complexity follows from a straightforward application of the classical shadows framework \cite{HKP20-classical-shadows}. See \cite[Corollary 21]{gross2006hudson} for a proof that there are $2^{\Theta(n^2)}$ many stabilizer states.}
Thus our algorithm offers a substantial improvement in the regime of $\tau = \omega(n^{-1/4})$. Arguably, the most interesting setting of parameters is constant $\tau$, in which case we have a quadratic improvement in sample complexity and a superpolynomial improvement in time complexity.

Observe that, because we output a witness of stabilizer fidelity at least $\tau - \eps$ with high probability, assuming a state with fidelity $\tau$ exists, our algorithm can additionally be used as a subroutine to \emph{estimate} stabilizer fidelity and, moreover, find a stabilizer state that witnesses this. 
More precisely, if the goal is to estimate stabilizer fidelity to accuracy $\pm \eps$, then one can break $[0,1]$ into intervals of width $\eps$ and perform a binary search procedure using our algorithm. Overall, this takes $O(n/\eps^6)$ samples and $\exp(O(n/\eps^4))$ time.

As an application, our stabilizer state approximation algorithm could be used to search for better stabilizer decompositions of magic states.
Recall that magic states are states that, when injected into Clifford circuits,
allow for %
universal quantum computation \cite{bravyi2005magicstates}. 
The best-known algorithms for simulating quantum circuits dominated by Clifford gates use decompositions of magic states into linear combinations of stabilizer states and have a runtime that scales polynomially in the complexity of the decomposition, either in terms of the stabilizer rank or stabilizer extent \cite{Bravyi2019simulationofquantum}. 
Hence, better stabilizer decompositions of magic states yield faster algorithms.
These decompositions are often obtained by writing the tensor product of a small number of magic states (usually on the order of $10$ qubits) as linear combination of a slightly larger number of stabilizer states \cite{bravyi2016trading, kocia2022improved}.
Therefore, if a classical simulation of our algorithm could be made practical for (say) $n \approx 15$ qubits, there is reason to believe that running this algorithm on magic states, combined with a meta-algorithm such as matching pursuit \cite{mallat1994matching}, could find better stabilizer decompositions of magic states and, as a result, improve the runtime of near-Clifford simulation.

Finally, we remark that the problem we solve is similar in spirit to the agnostic probably approximately correct (PAC) learning framework \cite{valiant1984learning, kearns1992toward}. 
In the agnostic PAC model, a learner is given labeled training data $\{(x_1, y_1), \ldots, (x_m, y_m)\}$ from some unknown distribution $\calD$, as well as some concept class $\mathcal{C}$ to choose a hypothesis from.
The goal of the learner is to find a hypothesis function $h \in \calC$ that approximates the best fit for the training data, even though no function in $\calC$ will necessarily fit the training data perfectly.
In an analogous fashion, our algorithm finds a stabilizer state $\ket{\phi}$ that approximates the best fit for $\ket{\psi}$ over the set of stabilizer states, which need not contain $\ket{\psi}$.
We note that Aaronson studied PAC learning of quantum states in the so-called realizable setting \cite{aaronson2007learning}. However, agnostic PAC learning of quantum states has not yet appeared in the literature. 

\subsubsection*{Bounded-Distance Stabilizer Approximation}

Although our stabilizer state approximation algorithm significantly improves upon brute force, it still requires exponential time in general. One might wonder whether this exponential runtime is necessary. For example, is it possible that finding stabilizer state approximations is computationally hard, even for states whose distance to the nearest stabilizer state is bounded by some small constant? \textit{A priori}, this might even be expected, because in other contexts, learning stabilizer states with a constant rate of noise can be as hard as the Learning Parities with Noise (LPN) problem \cite{Gollakota2022hardnessofpac,hinsche-single-t-gate}, which is believed to be hard.
What if the stabilizer fidelity is large enough to guarantee the existence of a \textit{unique} closest stabilizer state?
Our third result shows that in this regime, a modification of the algorithm from \cref{thm:fidelity_learning_informal} is computationally efficient. In particular, this modification works when the stabilizer fidelity is larger than $\cos^2(\pi/8) \approx 0.8536$, which is precisely threshold above which $\ket{\psi}$ is guaranteed to have a unique closest stabilizer state.

\begin{theorem}[Informal version of \cref{thm:bounded-distance}]
\label{thm:bounded-distance-informal}
    Fix $\gamma > 0$. There is an algorithm that, given copies of an $n$-qubit pure state $\ket{\psi}$ that has fidelity at least $\cos^2(\pi/8) + \gamma$ with some stabilizer state $\ket\phi$, returns $\ket\phi$ with high probability. The algorithm uses $O\left(n + \frac{\log n}{\gamma^2}\right)$ copies of $\ket{\psi}$ and $O\left(n^3 + \frac{n^2 \log n}{\gamma^2}\right)$ time.
\end{theorem}

Note that, unlike \cref{thm:fidelity_learning_informal}, this algorithm finds the stabilizer state $\ket{\phi}$ witnessing the maximum fidelity $F_\stabset(\ket{\psi})$, rather than a (possibly different) state witnessing fidelity $F_\stabset(\ket{\psi}) - \eps$.

\subsubsection*{Tolerant Stabilizer Testing}
Our final result is a tolerant property testing algorithm for stabilizer states. 
In the tolerant property testing model \cite{parnas2006tolerant}, which generalizes ordinary property testing \cite{rubinfeld1996robust, goldreich1998property}, 
a tester must accept objects that are at most $\eps_1$-far from having some property (``completeness'') and reject objects that are at least $\eps_2$-far from having that same property (``soundness'') for $0 \leq \eps_1 < \eps_2 \leq 1$. 
The standard property testing model is recovered when $\eps_1 = 0$, and the relaxed completeness condition generally makes tolerant testing a much harder problem. 
Nonetheless, the tolerant testing model is natural to consider in certain error models, such as in the presence of imprecise quantum gates.

Our result extends work by Gross, Nezami, and Walter \cite{gross2021schur}, who gave a property tester (hereafter, the ``GNW algorithm'') for stabilizer states.
When combined with the prior work of \cite{grewal_et_al:LIPIcs.ITCS.2023.64}, we deduce the existence of a \textit{tolerant} property testing algorithm for stabilizer states. Our algorithm takes copies of an $n$-qubit quantum state $\ket\psi$ and decides whether $\ket\psi$ has stabilizer fidelity at least $\alpha_1$ or at most $\alpha_2$, promised that one of these is the case. 
Note that we have taken $\alpha_1 \coloneqq 1 - \eps_1$ and $\alpha_2 \coloneqq 1-\eps_2$ for notational simplicity.

\begin{theorem}[Informal version of \cref{thm:tolerant-tester}]
Fix $\alpha_1,\alpha_2 \in [0,1]$ such that $\alpha_2 < \frac{4\alpha_1^6-1}{3}$, 
and define $\gamma \coloneqq \alpha_1^6 - \frac{3 \alpha_2+ 1}{4}$. 
There is an algorithm that uses $O(1/\gamma^2)$ copies of a quantum state $\ket\psi$, $O(n/\gamma^2)$ time, and decides whether $\ket\psi$ has stabilizer fidelity at least $\alpha_1$ or at most $\alpha_2$, promised that one of these is the case.
\label{thm:tolerant_tester_informal}
\end{theorem}

While our algorithm does not work for all settings of $\eps_1$ and $\eps_2$---giving such an algorithm is an open problem---our algorithm does significantly improve over prior work. 
In \cref{subsec:parameter-regime}, we compare the parameter regimes in which our algorithm works to the existing literature and show those regimes visually in \cref{fig:parameter-regime}.  

We remark that the algorithm from \cref{thm:bounded-distance-informal} can also be modified into a tolerant property testing algorithm that works whenever $\alpha_1 > \cos^2(\pi/8)$, by simply estimating the fidelity $\abs{\braket{\psi|\phi}}^2$ of the stabilizer state $\ket{\phi}$ output by the algorithm. The main advantages of \cref{thm:tolerant_tester_informal} are its improved sample complexity and runtime: whereas \cref{thm:tolerant_tester_informal} uses a system-size independent number of samples and linear time, \cref{thm:bounded-distance-informal} requires a linear number of samples and cubic time. Additionally, \cref{thm:tolerant_tester_informal} operates in some parameter regimes where $\alpha_1 \le \cos^2(\pi/8)$; see \cref{fig:parameter-regime}.

\subsection{Our Techniques}
The unifying tool in our work is \emph{Bell difference sampling}, a measurement primitive that has recently found applications in a variety of algorithms related to stabilizer states \cite{montanaro-bell-sampling,gross2021schur,grewal_et_al:LIPIcs.ITCS.2023.64}. We defer a full definition of Bell difference sampling to \cref{subsec:bell-difference-sampling}, but note some of its important properties here. Bell difference sampling involves measuring pairs of qubits of $\ket{\psi}^{\otimes 2}$ in the Bell basis, repeating again with $\ket{\psi}^{\otimes 2}$, and combining the measurements to interpret the result as corresponding to an $n$-qubit Pauli operator. Overall, this consumes four copies of $\ket{\psi}$, though it only performs measurements across two copies of $\ket{\psi}$ at a time. It will be most convenient to parameterize the sampled Pauli operators by strings in $\F_2^{2n}$, which we do as follows. For $x = (a,b) \in \F_2^{2n}$, where $a$ and $b$ are the first and last $n$ bits of $x$, respectively, we define the \emph{Weyl operator} $W_x$ as 
\[
W_x \coloneqq i^{a \cdot b} X^{a_1}Z^{b_1} \otimes \dots \otimes X^{a_n} Z^{b_n}.
\]
Importantly for us, the Weyl operators form an orthogonal basis for $\C^{2^n \times 2^n}$, and so they give rise to the \emph{Weyl expansion} of a quantum state $\ket{\psi}$ as
\[
\ketbra{\psi}{\psi} = \frac{1}{2^n} \sum_{x \in \F_2^{2n}} \braket{\psi|W_x|\psi}  W_x.
\]
For pure states, the squared coefficients in this expansion sum to $1$, and therefore form a distribution over $\F_2^{2n}$. We denote this distribution by $p_\psi(x) \coloneqq 2^{-n} \braket{\psi|W_x|\psi}^2$.\footnote{Here is an easy proof that $p_\psi$ is a distribution: $\sum_x p_\psi(x) = \sum_x 2^{-n}\braket{\psi|W_x|\psi}^2 = \tr(\ket\psi\!\!\braket{\psi|\psi}\!\!\bra\psi) = 1$, where the second step follows from using the Weyl expansion of $\ket{\psi}$.\label{footnote:p-psi-dist}}

Gross, Nezami, and Walter \cite{gross2021schur} give an explicit form for the distribution obtained by performing Bell difference sampling. In particular, they showed that Bell difference sampling a quantum pure state $\ket{\psi}$ is equivalent to sampling from the following distribution:
\[
q_\psi(x) \coloneqq \sum_{a \in \F_2^{2n}} p_\psi(a) p_\psi(a + x),
\]
i.e., the convolution of $p_\psi$ with itself. 
At a high level, we establish our results by proving some structure on $q_\psi$ and $p_\psi$ for certain quantum states. 

\paragraph{Tight Pseudorandomness Bounds}
To prove our lower bound on the number of non-Clifford gates required to prepare pseudorandom states, we give an algorithm that distinguishes Haar-random states from quantum states prepared by circuits with fewer than $n/2$ non-Clifford single-qubit gates. 
The key insight is that if $\ket \psi$ is the output of such a circuit, then $q_\psi$ is concentrated on a proper subspace of $\F_2^{2n}$, whereas for Haar-random states, $q_\psi$ is anticoncentrated on all such subspaces with overwhelming probability over the Haar measure. 
Proving these properties of $q_\psi$ reveals a simple algorithm: draw a linear number of samples from $q_\psi$ and compute the number of linearly independent vectors in the sample.  Haar-random states will have $2n$ such vectors with high probability and, otherwise, there will be strictly less than $2n$ such vectors. 

\paragraph{Faster Stabilizer State Approximation}
Our algorithms for stabilizer approximation also rely on proving anticoncentration properties of $q_\psi$. We begin by showing that if $\ket{\psi}$ has large fidelity with some stabilizer state $\ket{\phi}$, then $q_\psi$ is well-supported on the $n$-dimensional subspace $\weyl(\ket{\phi}) \coloneqq \{x \in \F_2^{2n} : W_x\ket{\phi} = \pm \ket{\phi} \}$
of Weyl operators that stabilize $\ket{\phi}$ (up to sign). Next, we establish that if $\ket{\phi}$ is the state that maximizes stabilizer fidelity, then the mass of $q_\psi$ on $\weyl(\ket{\phi})$ is not too concentrated on any proper subspace. 
Hence, by sampling from $q_\psi$ enough times, we can be guaranteed that with high probability, $\weyl(\ket{\phi})$ will be generated by some subset of the sampled Weyl operators. 
By iterating through all mutually commuting subsets of the sampled Weyl operators, we compile a list of candidate stabilizer states $\ket{\phi}$ that must contain the fidelity-maximizing $\ket{\phi}$. 
Therefore, our algorithm reduces to estimating the fidelity of $\ket{\psi}$ with each candidate $\ket{\phi}$. We further improve the time efficiency via an algorithm for finding maximal cliques, due to \cite{tomita2006worst}, by observing that the candidate subsets must correspond to maximal cliques in the graph of commutation relations.\footnote{I.e., the graph whose edges connect nodes corresponding to commuting Weyl operators.} We also improve the sample complexity by using the classical shadows protocol \cite{HKP20-classical-shadows} to estimate all of the fidelities with candidate states $\ket{\phi}$ efficiently. For more details on these improvements, see \cref{subsec:stab-fidelity-algo}.

\paragraph{Bounded-Distance Stabilizer Approximation}
In the case where stabilizer fidelity is bounded below by $\cos^2(\pi /8)$, we follow the same approach, but use a different and more efficient subroutine for determining which of the sampled Weyl operators generate $\weyl(\ket \phi)$. In particular, we show that there is a simple statistical test for this purpose: if $\abs{\braket{\phi|\psi}}^2 > \cos^2(\pi/8)$, then for any $x \in \F_2^{2n}$, $x \in \weyl(\ket\phi)$ if and only if $\braket{\psi | W_x | \psi}^2 > \frac{1}{2}$ (\cref{prop:always-larger}).
This allows us to eschew the maximal clique algorithm entirely, and we instead directly estimate $\braket{\psi | W_x | \psi}^2$ to determine whether $W_x$ belongs to $\weyl(\ket\phi)$. We further improve upon the sample complexity of this subroutine by making use of an algorithm due to Huang, Kueng, and Preskill \cite{huang2021information} for estimating the expectation of $m$ different Weyl operators from only $O(\log m)$ samples. We also provide a simpler proof of this result in \cref{appendix:hkp}, based on the Fourier-analytic techniques described below.

\paragraph{Tolerant Stabilizer Testing}
Our last result, the tolerant property testing algorithm for stabilizer states, is based on running the GNW stabilizer testing algorithm \cite{gross2021schur} repeatedly to estimate its acceptance probability. We prove this algorithm's correctness by combining existing bounds on the completeness and soundness of the GNW test in terms of stabilizer fidelity. The bound on completeness is due to \cite{grewal_et_al:LIPIcs.ITCS.2023.64}, while the soundness analysis comes from \cite{gross2021schur}.

\paragraph{Symplectic Fourier Analysis}

An essential tool for proving the above results is symplectic Fourier analysis, wherein the Fourier transform over real-valued functions is defined with respect to the symplectic product on $\F_2^{2n}$.
To give a sense of the usefulness of symplectic Fourier analysis in our work, we showcase two powerful theorems whose proofs are symplectic-Fourier-analytic.
In what follows, for a subspace $T \subseteq \F_2^{2n}$ identified with a set of Weyl operators $\{W_x : x \in T\}$, the subspace $T^\sympcomp$ denotes the set of Weyl operators that commute with $T$.

\begin{theorem}[Restatement of \cref{thm:p-mass-identity-subgroups}
and \cref{thm:q-mass-identity-subgroups}]
\label{thm:p_duality_intro}
Let $T \subseteq \F_2^{2n}$ be a subspace, and let $\ket    \psi$ be an $n$-qubit quantum pure state.  Then 
\[
\sum_{a \in T}p_\psi(a) =  \frac{\abs{T}}{2^n}\sum_{x \in T^{\sympcomp}}p_\psi(x),   
\]
and
\[
\sum_{a \in T}q_\psi(a) = \abs{T}  \sum_{x \in T^{\sympcomp}}p_\psi(x)^2.  
\]
\end{theorem}

In words, \cref{thm:p_duality_intro} shows that $p_\psi$ and $q_\psi$ exhibit a strong duality property with respect to the commutation relations among Weyl operators. In particular, the first part shows that the mass of $p_\psi$ on a subspace $T$ of Weyl operators is directly proportional to the mass on the subspace $T^\sympcomp$ of Weyl operators that commute with $T$.
\cref{thm:p_duality_intro} is especially powerful when the subspace $T$ is very large, because $T$ and $T^\perp$ always have inversely proportional size (see \cref{fact:symp-vector-space-facts}). Hence, using our duality theorems, we can convert summations over high-dimensional subspaces into summations over just a few terms.

\section{Preliminaries}
We introduce notation and background that is central to our work. 
We assume familiarity with common concepts in quantum information and computer science, such as the stabilizer formalism and basic graph theory. 
For more background on the stabilizer formalism, see, e.g., \cite{Got97-thesis,nielsen2002quantum}.

We write $[n] \coloneqq \{1, \ldots, n\}$. 
For $x = (a,b) \in \F_2^{2n}$, $a$ and $b$ always denote the first and last $n$ bits of $x$, respectively.
For a probability distribution $\calD$ on a set $S$, we denote drawing a sample $s \in S$ according to $\calD$ by $s \sim \calD$. 
We denote drawing a sample $s \in S$ uniformly at random by $s \sim S$.
In an undirected graph $G$, a clique is a complete subgraph of $G$. 
A maximal clique is a clique that is not a proper subgraph of another clique. 
For quantum pure states $\ket\psi, \ket\phi$, let $\tracedistance{\ket\psi,\ket\phi} = \sqrt{1 - \abs{\braket{\psi|\phi}}^2}$ denote the trace distance
and $F(\ket\psi, \ket\phi) = \abs{\braket{\psi|\phi}}^2$ denote the fidelity.
The trace distance quantifies the distinguishability between two quantum states by a two-outcome measurement. 
We also use the following Chernoff bound. 

\begin{fact}[Chernoff bound]\label{fact:chernoff}
Let $X_1, \ldots, X_n$ be independent identically distributed random variables taking values in $\{0,1\}$. Let $X$ denote their sum and let $\mu = \E[X]$. Then for any $\delta > 0$, 
\[
\Pr\left[X \leq (1- \delta)\mu\right] \leq e^{-\delta^2 \mu / 2}.
\]
\end{fact}

We additionally require the following version of Hoeffding's inequality. 

\begin{fact}[Hoeffding’s inequality]\label{fact:hoeffding}
Suppose $X_1,\dots, X_n$ are independent random variables subject to $a_i \le X_i \le b_i$ for all $i$. 
Let $X = \sum_{i=1}^n X_i$ and let $\mu = \E[X]$. 
Then for all $t \ge 0$ it holds that:
\[
\Pr[X - \mu \geq t] \le \exp\left(- \frac{2t^2}{\sum_{i=1}^n (b_i - a_i)^2} \right)
\]
and
\[
\Pr[\abs{X - \mu} \ge t] \le 2 \exp\left(- \frac{2t^2}{\sum_{i=1}^n (b_i - a_i)^2} \right).
\]
\end{fact}

The $n$-qubit Pauli group $\mathcal{P}_n$ is the set $\{\pm 1, \pm i\} \times \{I, X, Y, Z\}^{\otimes n}$, where $I, X, Y, Z$ are the standard Pauli matrices. We refer to unitary transformations in the Clifford group as Clifford circuits (equivalently, Clifford circuits are quantum circuits comprised only of Clifford \emph{gates}, namely, the Hadamard, Phase, and CNOT gates).
Clifford gates with the addition of any single-qubit non-Clifford gate form a universal gate set. The $T$-gate is often the non-Clifford gate of choice,  where the $T$-gate is defined by $T \coloneqq \ket{0}\!\!\bra{0} + e^{i \pi/4} \ket{1}\!\!\bra{1}$.
We denote the set of $n$-qubit stabilizer states by $\stabset_n$.
One way to measure the ``stabilizer complexity'' of a quantum state is the stabilizer fidelity. 

\begin{definition}[Stabilizer fidelity, {\cite[Definition 4]{Bravyi2019simulationofquantum}}]
Suppose $\ket{\psi}$ is a pure $n$-qubit state.
The \emph{stabilizer fidelity} of $\ket{\psi}$, denoted $\fidelity_{\stabset}$, is 
\[
\fidelity_{\stabset}(\ket\psi) \coloneqq \max_{\ket{\phi} \in \stabset_n}\abs{\braket{\phi|\psi}}^2.
\]
\end{definition}

\subsection{Symplectic Vector Spaces}\label{subsec:symplectic-vector}

We work extensively with $\F_2^{2n}$ as a symplectic vector space by equipping it with the symplectic product.

\begin{definition}[Symplectic product]\label{def:symplectic-product}
For $x,y \in \F_2^{2n}$, we define the \emph{symplectic product} as $[x,y] = x_1 \cdot y_{n+1} + x_2\cdot y_{n+2} + ... + x_n \cdot y_{2n} + x_{n+1} \cdot y_1 + x_{n+2} \cdot y_2 + ... + x_{2n} \cdot y_n$, where all operations are performed over $\F_2$.
\end{definition}

The symplectic product gives rise to the notion of a \emph{symplectic complement}, much like the orthogonal complement for the standard inner product.

\begin{definition}[Symplectic complement]\label{def:symplectic-complement}
Let $T \subseteq \F_2^{2n}$ be a subspace. The \emph{symplectic complement} of $T$, denoted by $T^\sympcomp$, is defined by
\[
T^\sympcomp \coloneqq \{a \in \F_2^{2n} : \forall x \in T,\, [x,a] = 0 \}.
\]
\end{definition}

We present the following useful facts about the symplectic complement, many of which are similar to that of the more familiar orthogonal complement.
\begin{fact}\label{fact:symp-vector-space-facts}
Let $S$ and $T$ be subspaces of $\F_2^{2n}$. Then:
\begin{itemize}
\item $T^\sympcomp$ is a subspace.
\item $(T^\sympcomp)^\sympcomp = T$.
\item $\abs{T} \cdot \abs{T^\sympcomp} = 4^n$, or equivalently $\dim T + \dim T^\sympcomp = 2n$.
\item $T \subseteq S \iff S^\sympcomp \subseteq T^\sympcomp$.
\end{itemize}
\end{fact}

A subspace $T \subset \F_2^{2n}$ is \emph{isotropic} when for all $x,y \in T$, $[x,y] = 0$.
A subspace $T \subset \F_2^{2n}$ is \emph{Lagrangian} when $T^\sympcomp = T$. 
Lagrangian subspaces can equivalently be defined as isotropic subspaces with dimension $n$.

\subsection{Symplectic Fourier Analysis}\label{subsec:symplectic-fourier}

Our work uses \emph{symplectic} Fourier analysis, which is similar to Boolean Fourier analysis (see e.g., \cite{o2014analysis}), except the Fourier characters are defined with respect to the symplectic product.

\begin{definition}[Symplectic Fourier transform]
Let $f:\F_2^{2n} \to \R$. We define the \emph{symplectic Fourier transform} of $f$, which is given by a function $\hat{f}:\F_2^{2n} \to \R$, by
\[
\hat{f}(a) = \dfrac{1}{4^n} \sum_{x \in \F_2^{2n}} (-1)^{[a,x]}f(x).
\]
Hence, the \emph{symplectic Fourier expansion} of $f$ is
\[
f(x) = \sum_{a \in \F_2^{2n}} (-1)^{[a,x]}\hat{f}(a).
\]
\end{definition}

The well-known Plancherel's Theorem holds under this Fourier transform. 
\begin{fact}[Plancherel's Theorem]\label{fact:plancherel}
\[
\frac{1}{4^n} \sum_{x \in \F_2^{2n}} f(x) g(x) = \sum_{x \in  \F_2^{2n}} \widehat{f}(x)\widehat{g}(x).
\]
\end{fact}
\begin{proof}
    \begin{align*}
        \frac{1}{4^n}\sum_{x \in  \F_2^{2n}} f(x) g(x)
        &= \frac{1}{4^n}\sum_{a,b \in  \F_2^{2n}} \widehat{f}(a) \widehat{g}(b) \sum_{x \in  \F_2^{2n}} (-1)^{[b + c, x]}\\
        &= \frac{1}{4^n}\sum_{a,b \in  \F_2^{2n}} \widehat{f}(a) \widehat{g}(b)  \left( 4^n \indic{a = b}\right)\\
        &= \sum_{x \in  \F_2^{2n}} \widehat{f}(x) \widehat{g}(x) && \qedhere
    \end{align*}
\end{proof}

Convolution plays an important role in our work.

\begin{definition}[Convolution]
Let $f, g : \F_2^{2n} \to \R$. Their convolution is the function $f \ast g:\F_2^{2n} \to \R$ defined by 
\[
(f \ast g)(x) = \E_{t \sim \F_2^{2n}}[f(t) g(t + x)] = \frac{1}{4^n} \sum_{t \in \F_2^{2n}} f(t) g(t + x).
\]
\end{definition}

Convolution corresponds to the multiplication of Fourier coefficients, even under the symplectic Fourier transform. 

\begin{proposition}\label{thm:convolution-theorem}
   Let $f, g: \F_2^{2n} \to \R.$ Then for all $a \in \F_2^{2n}$, 
   \[
   \wh{f \ast g}(a) = \wh{f}(a) \wh{g}(a).
   \]
\end{proposition}

\begin{proof}
A useful observation is that the symplectic product is bilinear, such that $[a,x] = [a,t] + [a,x+t]$. Using this, we can expand and simplify:
\begin{align*}
\wh{f \ast g}(a) &= \dfrac{1}{4^n}\sum_{x \in \F_2^{2n}} (-1)^{[a,x]} (f \ast g)(x) \\
&= \dfrac{1}{16^n}\sum_{x,t \in \F_2^{2n}} (-1)^{[a,x]} f(t)g(x+t) \\
&= \dfrac{1}{16^n}\sum_{t \in \F_2^{2n}} (-1)^{[a,t]} f(t)
\sum_{x \in \F_2^{2n}}(-1)^{[a,x+t]} g(x+t) \\
&= \wh{f}(a)\wh{g}(a). \qedhere
\end{align*}
\end{proof}

We prove a fact that will be useful in our symplectic Fourier analysis.

\begin{lemma}\label{lemma:sum-over-characters}
For any subspace $T \subseteq \F_2^{2n}$ and a fixed $x \in \F_2^{2n}$,
    \[\sum_{a \in T} (-1)^{[a, x]} = \abs{T} \cdot \indic{x \in T^\perp}.\]
\end{lemma}
\begin{proof}
    If $x \in T^\sympcomp$ then this is easy to see. 
    Suppose $x \not \in T^\sympcomp$. Then we claim $[a, x] = 0$ for exactly half of the elements $a \in T$. 
    To see this, we observe that there exists a $y \in T$ such that $[y,x] = 1$. 
    Let $T / y$ denote $T$ modulo addition by $y$. 
    Given a pair $\{a, a + y\} \in T / y$, observe that exactly one of $[a, x]$ and $[a+y, x]$ is $0$ and the other is $1$.
    As such we have that for half of all $a \in T$, $[a,x] = 0$ and for the other half, $[a,x] = 1$, giving us $\sum_{a \in T} (-1)^{[a, x]} = 0$. \qedhere
\end{proof}

\subsection{Weyl Operators and Bell Difference Sampling}
\label{subsec:bell-difference-sampling}
For $x = (a, b) \in \F_2^{2n}$, the \emph{Weyl operator} $W_x$ is defined as 
\[
W_x \coloneqq 
i^{a'\cdot b'}(X^{a_1} Z^{b_1}) \otimes \dots \otimes (X^{a_n} Z^{b_n}),
\]
where $a',b' \in \Z^n$ are the embeddings of $a,b$ into $\Z^n$.
Each Weyl operator is a Pauli operator, and every Pauli operator is a Weyl operator up to a phase.
Because the Clifford group normalizes the Pauli group, Clifford circuits induce an action on $\F_2^{2n}$ by conjugation of the corresponding Weyl operators (up to phase). That is, for every Clifford circuit $C$ and $x \in \F_2^{2n}$, there exists a unique $y \in \F_2^{2n}$ and phase $\alpha \in \{ \pm 1\}$ such that $C W_x C^\dagger = \alpha W_y$. In a slight abuse of notation, we denote this action on $\F_2^{2n}$ by $C(x) = y$.

There is clearly a bijection between $\F_2^{2n}$ and the set of Weyl operators, so any subset of $\F_2^{2n}$ corresponds to a subset of Weyl operators.
Importantly, commutation relations between Weyl operators are determined by the symplectic product. 
In particular, for $x,y \in \F_2^{2n}$, the Weyl operators $W_x, W_y$ commute when $[x,y] = 0$ and anticommute when $[x,y] = 1$.
So, if $T \subseteq \F_2^{2n}$ is a subspace, then $T$ is isotropic if and only if $\{W_x : x \in T\}$ is a set of mutually commuting Weyl operators.
Similarly, $T$ is Lagrangian if and only if $\{W_x : x \in T\}$ is a set of $2^n$ mutually commuting Weyl operators. 

\begin{definition}[Unsigned stabilizer group]
Let $\weyl(\ket\psi) \coloneqq \{x \in \F_2^{2n} : W_x \ket{\psi} = \pm \ket{\psi}\}$ denote the unsigned stabilizer group of $\ket{\psi}$.
\end{definition}

It is not hard to show that, as a consequence of the uncertainty principle, $\weyl(\ket\psi)$ is an isotropic subspace of $\F_2^{2n}$. 
Additionally, if $T \subset \F_2^{2n}$ is a Lagrangian subspace, then the set of states $\{\ket{\varphi} : \weyl(\ket{\varphi}) = T\}$ forms an orthonormal basis of the $n$-qubit Hilbert space. Moreover, since each basis state $\ket{\varphi}$ is stabilized by $2^n$ Weyl operators (up to phase), every basis state is a stabilizer state. Conversely, observe that for any stabilizer state $\ket{\varphi}$, $\weyl(\ket{\varphi})$ is a Lagrangian subspace.

We now define a new stabilizer complexity measure based on the unsigned stabilizer group. 

\begin{definition}[Stabilizer dimension]\label{def:stabilizer-dimension}
Let $\ket{\psi}$ be an $n$-qubit pure state. The \emph{stabilizer dimension of $\ket{\psi}$} is the dimension of $\weyl(\ket{\psi})$ as a subspace of $\F_2^{2n}$.\footnote{The stabilizer dimension is closely related to the stabilizer nullity \cite{beverland2020lower} (in fact, for $n$-qubit states, the stabilizer dimension is simply $n$ minus the stabilizer nullity).  }
\end{definition}

The stabilizer dimension of a stabilizer state is $n$, which is maximal, and, for most states, the stabilizer dimension is $0$.

The Weyl operators collectively form an orthogonal basis for $2^n \times 2^n$ matrices with respect to the inner product $\langle A, B \rangle = \tr{(A^\dagger B)}$. This gives rise to the so-called \emph{Weyl expansion} of a quantum state. 

\begin{definition}[Weyl expansion]\label{def:weyl-expansion}
Let $\ket{\psi} \in \C^{2^n}$ be an $n$-qubit quantum pure state. The Weyl expansion of $\ket\psi$ is 
\[
\ketbra{\psi}{\psi} = \dfrac{1}{\sqrt{2^n}} \sum_{x \in F_2^{2n}} c_\psi(x) W_x,
\]
where $c_\psi(x) = \frac{1}{\sqrt{2^n}}  \braket{\psi|W_x|\psi}$. 
\end{definition}

Squaring the $c_\psi(x)$'s gives rise to a distribution over $\F_2^{2n}$ and therefore over the Weyl operators (see \cref{footnote:p-psi-dist} for a proof).
We denote this distribution by $p_\psi(x) \coloneqq c_\psi(x)^2$ and refer to it as the \emph{characteristic distribution}. Note that, for all $x$, $p_\psi(x) \in [0, 2^{-n}]$.   
A convenient fact about the $p_{\psi}$ is its invariance (up to scaling) under the symplectic Fourier transform.
\begin{fact}
\label{fact:weyl-fourier-duality}
For any $n$-qubit pure state $\ket\psi$ and any $x \in \F_2^{2n}$, $p_\psi(x) = 2^n\hat{p_{\psi}}(x)$.
\end{fact}
For a proof of this fact, we refer the reader to \cite[Equation 3.5]{gross2021schur}, noting our slight difference in normalization.\footnote{Alternatively, one can refer to \cite[Proposition 17]{grewal_et_al:LIPIcs.ITCS.2023.64}, where the normalization is consistent with this work, but \cite{grewal_et_al:LIPIcs.ITCS.2023.64} uses the standard Fourier transform rather than the symplectic one. Despite this difference, the proof goes through in a similar way.} 

A significant algorithmic primitive in our work is \emph{Bell difference sampling} \cite{montanaro-bell-sampling, gross2021schur}.
Let $\ket{\Phi^+} \coloneqq \frac{\ket{00} + \ket{11}}{\sqrt{2}}$. 
Then, the set of quantum states $\{\ket{W_x} \coloneqq (W_x \otimes I) \ket{\Phi^+} : x\in\mathbb{F}_2^{2}\}$ forms an orthonormal basis of $\mathbb{C}^{2} \otimes \mathbb{C}^{2}$, which we call the \textit{Bell basis}. 
Bell difference sampling an $n$-qubit state $\ket{\psi}$ just means the following.
First, take two copies of a pure state $\ket{\psi}$. 
Take the first qubit in each copy and measure them in the Bell basis. 
Repeat this for each remaining pair of qubits. 
Let $(a_i, b_i)$ denote the two-bit measurement outcome from measuring the $i$th pair of qubits. 
Then, we denote the measurement outcome on the two copies by $x = (a_1, \ldots, a_n, b_1, \ldots, b_n) \in \F_2^{2n}$. 
Repeat this once more with two fresh copies of $\ket{\psi}$ to obtain a string $y \in \F_2^{2n}$. 
Finally, output $x + y$.\footnote{Even when $\ket{\psi}$ is a stabilizer stabilizer state, measuring two copies of $\ket{\psi}$ in the Bell basis returns $x \in \F_2^{2n}$ with probability $p_\psi(x+a)$, where $a \in \F_2^{2n}$ is an unwanted shift. Bell difference sampling essentially cancels out this unwanted shift $a$. See \cite{montanaro-bell-sampling, gross2021schur} for more detail.} 
Historically, Bell difference sampling has found use in algorithms for stabilizer states. However, Gross, Nezami, and Walter proved that Bell difference sampling is meaningful for all quantum states.

\begin{lemma}[Bell difference sampling, {\cite[Theorem 3.2]{gross2021schur}}]\label{lem:bell_diff_sampling}
Let $\ket{\psi}$ be an arbitrary $n$-qubit pure state. 
Bell difference sampling corresponds to drawing a sample from the following distribution:
\[
q_{\psi}(x) \coloneqq 4^n (p_\psi \ast p_\psi)(x) = \sum_{y \in \F_2^{2n}} p_{\psi}(y) p_{{\psi}}(x + y),
\]
and uses four copies of $\ket\psi$.
We refer to $q_\psi(x)$ as the \emph{Weyl distribution}.
\end{lemma}

\section{On the Weyl and Characteristic Distributions}\label{sec:duality}

We prove identities related to the characteristic distribution $p_\psi$ and Weyl distribution $q_\psi$ that are critical for our results. 
We emphasize that these results hold for \emph{all} pure quantum states.
First, we show that the mass on a subspace $T \subseteq \F_2^{2n}$ under $p_\psi$ is proportional to the mass on $T^\sympcomp$ under $p_\psi$.

\begin{theorem}\label{thm:p-mass-identity-subgroups}
Let $T \subseteq \F_2^{2n}$ be a subspace. Then 
\[
\sum_{a \in T}p_\psi(a) =  \frac{\abs{T}}{2^n}\sum_{x \in T^{\sympcomp}}p_\psi(x).  
\]
\end{theorem}
\begin{proof}
    \begin{align*}
        \sum_{a \in T}p_\psi(a) &= \sum_{a \in T} \sum_{x \in \F_2^{2n}} \widehat{p}_\psi(x) (-1)^{[a, x]}\\
        &= \frac{1}{2^n} \sum_{a \in T} \sum_{x \in \F_2^{2n}} p_\psi(x) (-1)^{[a, x]} && (\text{\cref{fact:weyl-fourier-duality}})\\
        &= \frac{\abs{T}}{2^n} \sum_{x \in \F_2^{2n}} p_\psi(x) \cdot \indic{x \in T^\perp} && \text{(\cref{lemma:sum-over-characters})}\\
        &= \frac{\abs{T}}{2^n} \sum_{x \in T^\perp} p_\psi(x). && \qedhere
    \end{align*} 
\end{proof}

A similar result is true for $q_\psi$. In words, we show that the average probability mass on a subspace $T$ under $q_\psi$ is equal to the squared-$\ell_2$-norm of the probability mass on $T^\sympcomp$ under $p_\psi$. 

\begin{theorem}\label{thm:q-mass-identity-subgroups}
Let $T \subseteq \F_2^{2n}$ be a subspace. Then 
\[
\frac{1}{\abs{T}} \sum_{a \in T}q_\psi(a) =  \sum_{x \in T^{\sympcomp}}p_\psi(x)^2.  
\]
\end{theorem}
\begin{proof}
\begin{align*}
    \sum_{a \in T} q_\psi(a) 
    &= \sum_{a \in T} \sum_{x \in \F_2^{2n}} \wh{q}(x) (-1)^{[a,x]} \\
    &= 4^n \sum_{a \in T} \sum_{x \in \F_2^{2n}} \wh{p}(x)^2 (-1)^{[a,x]} && \text{(\cref{lem:bell_diff_sampling}, \cref{thm:convolution-theorem}.)} \\
    &= \sum_{a \in T} \sum_{x \in \F_2^{2n}} p(x)^2 (-1)^{[a,x]} && \text{(\cref{fact:weyl-fourier-duality})} \\
    &= \abs{T} \sum_{x \in T^\sympcomp} p(x)^2. && \text{(\cref{lemma:sum-over-characters})}\qedhere  
\end{align*}
\end{proof}

\section{Pseudorandomness Lower Bounds}
\label{sec:pseudorandomness}
We prove that the output state of any Clifford circuit augmented with fewer than $n/2$ non-Clifford single-qubit gates can be efficiently distinguished from Haar random.\footnote{If we fix the non-Clifford gate to be a $T$-gate, then $n/2$ can be improved to $n$.} 
As a result, any circuit family that prepares an ensemble of $n$-qubit pseudorandom quantum states must use at least $\Omega(n)$ non-Clifford single-qubit gates. 
The key idea is that Haar-random states have minimal stabilizer dimension (\cref{def:stabilizer-dimension}) with overwhelming probability. By contrast, for a quantum circuit that acts on a stabilizer state (which has stabilizer dimension $n$), each single-qubit non-Clifford gate decreases the stabilizer dimension by at most $2$.

We introduce the following definition to simplify the exposition, borrowing terminology from \cite{leone-stabilizer-nullity}. 

\begin{definition}[$t$-doped Clifford circuits]
A $t$-doped Clifford circuit is a quantum circuit comprised only of Clifford gates (i.e., Hadamard, Phase, and $\mathrm{CNOT}$) and at most $t$ single-qubit non-Clifford gates that starts in the state $\ket{0^n}$.
\end{definition}

\subsection{Quantum Circuits With Few Non-Clifford Gates}
To begin, we  show that the output state $\ket \psi$ of a $t$-doped Clifford circuit, where $t < n/2$, induces a distribution $q_{\psi}$ that is supported over a subspace of dimension at most $2n-2$.

\begin{lemma}
\label{lem:arbitrary-gate-dimension}
    Let $\ket\psi$ be the output state of a $t$-doped Clifford circuit. Then the stabilizer dimension of $\ket\psi$ is at least $n-2t$.
\end{lemma}

\begin{proof}
    We proceed by induction on $t$. 
    In the base case $t=0$, so $\ket{\psi}$ is a stabilizer state and has stabilizer dimension $n$.

    For the inductive step, let $t > 0$. Write $\ket{\psi} = CU\ket{\varphi}$, where $\ket{\varphi}$ is the output of a $(t-1)$-doped Clifford circuit, 
    $U$ is a single-qubit gate, and $C$ is a Clifford circuit. Because the stabilizer dimension is unchanged by Clifford gates, it suffices to show that the stabilizer dimension of $U\ket{\varphi}$ is at least $n - 2t$.

    Let $S = \weyl(\ket{\varphi})$, which by the induction assumption has dimension at least $n - 2(t-1)$. Observe that for any $x \in S$, if the Weyl operator $W_x$ commutes with $U$, then:
    \[
    \bra{\varphi}U^\dagger W_x U \ket{\varphi} = \bra{\varphi} W_x \ket{\varphi} = \pm 1.
    \]
    Hence, letting $T \coloneqq \{x \in S: UW_xU^\dagger = W_x \}$, we see that the stabilizer dimension of $U\ket{\varphi}$ is at least the dimension of $T$. But $|T| \ge |S| / 4$, because $T$ contains all elements $x$ of $S$ for which $W_x$ restricts to the identity on the qubit to which $U$ is applied. Thus, the stabilizer dimension of $U\ket{\varphi}$ is at least $n - 2t$, as desired.
\end{proof}

We remark that the stabilizer dimension lower bound in \cref{lem:arbitrary-gate-dimension} can be improved to $n - t$ in the case that all of the non-Clifford gates are diagonal (for example, if all of the non-Clifford gates are $T$-gates). This is because diagonal gates commute with both $I$ \textit{and} $Z$.

\begin{lemma}
\label{lem:p-support-sympcomp}
    The support of $p_\psi$ is contained in $\weyl(\ket\psi)^\sympcomp$. 
\end{lemma}
\begin{proof}
    We show the mass of $p_{\psi}$ on $\weyl(\ket\psi)^\sympcomp$ is $1$. 
    \begin{align*}
    \sum_{x \in \weyl(\ket\psi)^\sympcomp} p_\psi(x) 
    &= \frac{\abs{\weyl(\ket\psi)^\sympcomp}}{2^n}\sum_{x \in \weyl(\ket\psi)} p_\psi(x) && \text{(\cref{thm:p-mass-identity-subgroups})}\\
    &= \frac{\abs{\weyl(\ket\psi)^\sympcomp}}{2^n} \frac{\abs{\weyl(\ket\psi)}}{2^n} && \text{(By definition of $\weyl(\ket\psi)$)}\\
    &= 1. && (\text{\cref{fact:symp-vector-space-facts}}) \qedhere
    \end{align*} 
\end{proof}

\begin{corollary}
\label{cor:q-support-sympcomp}
The support of $q_\psi(x)$ is $\weyl(\ket\psi)^\sympcomp$.
\end{corollary}
\begin{proof}
    Suppose $x \not \in \weyl(\ket\psi)^\sympcomp$.
    We want to show that $q_\psi(x) = 0$. %
    By the definition of $q_\psi$ and by \cref{lem:p-support-sympcomp},
    \[
    q_\psi(x) = \sum_{a \in \F_2^{2n}} p_\psi(a)p_\psi(x + a) = \sum_{a \in \weyl(\ket\psi)^\sympcomp} p_\psi(a) p_\psi(x+a),
    \]
because $p_\psi(a) = 0$ for $a \not\in \weyl(\ket\psi)^\sympcomp$. In the right-most sum, since $a \in \weyl(\ket\psi)^\sympcomp$, $x + a \notin \weyl(\ket\psi)^\sympcomp$ if and only if $x \not\in \weyl(\ket\psi)^\sympcomp$. So, applying \cref{lem:p-support-sympcomp} again, $p_\psi(x + a) = 0$ for each term in the sum, implying that the total sum is $0$. \qedhere
\end{proof}

\begin{corollary}
\label{cor:q-support-dimension}
    Let $\ket\psi$ be the output state of a $t$-doped Clifford circuit. Then the support of $q_{\psi}$ is a subspace of dimension at most $n+2t$.
\end{corollary}
\begin{proof}
     By \Cref{lem:arbitrary-gate-dimension}, the dimension of $\weyl(\ket\psi)$ is at least $n-2t$, implying the dimension of $\weyl(\ket\psi)^\sympcomp$ is at most $n+2t$. The result follows from \Cref{cor:q-support-sympcomp}. 
\end{proof}

\subsection{Anticoncentration of Haar-Random States}
Now we show that if $\ket\psi$ is Haar-random, then $q_\psi$ is well-supported over the entirety of $\F_2^{2n}$ in the sense that every proper subspace of $\F_2^{2n}$ contains a bounded fraction of the $q_{\psi}$ mass. This implies that sampling from $q_{\psi}$ gives $2n$ linearly independent elements of $\F_2^{2n}$ after a reasonable number of iterations.

We first require the following lemma, which shows that the Weyl measurements are concentrated around $0$. Proved in \cite{grewal_et_al:LIPIcs.ITCS.2023.64}, this is a consequence of L\'{e}vy's lemma.

\begin{lemma}[{\cite[Corollary 22]{grewal_et_al:LIPIcs.ITCS.2023.64}}]
\label{lem:all_weyl_concentration}
Let $\ket{\psi}$ be a Haar-random $n$-qubit state. Then
\[
\Pr\left[\exists x \neq 0 : \abs{\braket{\psi | W_x | \psi}} \ge \eps \right] \le 2^{2n + 1}\exp\left(-\frac{2^n \eps^2}{36 \pi^3} \right).
\]
\end{lemma}

Combining with the fact (\cref{thm:q-mass-identity-subgroups}) that the $q_\psi$ mass on a subspace is proportional to its $p_\psi^2$ mass on the symplectic complement, we obtain the following.

\begin{lemma}
    \label{lem:proper-subspace-concentration}
    Let $\ket{\psi}$ be a Haar-random $n$-qubit state. Then all subspaces $T \subseteq \F_2^{2n}$ of dimension $2n-1$ simultaneously satisfy
    \[
    \sum_{x \in T} q_{\psi}(x) \leq \frac{2}{3},
    \]
    except with probability at most
    \[
    2^{2n + 1} \exp \left(-\frac{2^n}{36 \sqrt{3} \pi^3} \right).
    \]
\end{lemma}

\begin{proof}
    Let $T$ be any subspace of dimension $2n-1$. Then the symplectic complement $T^\perp$ has dimension $1$, so it is the span of a single nonzero $x \in \F_2^{2n}$. By \cref{thm:q-mass-identity-subgroups},
    \[
    \sum_{a \in T} q_\psi(a) = 2^{2n-1}\sum_{a \in T^\sympcomp} p_\psi(a)^2 = \frac{1 + \braket{\psi|W_x|\psi}^4}{2}.
    \]
    Hence, the probability that there exists a $T$ for which $\sum_{x \in T} q_\psi(x)$ exceeds $\frac{2}{3}$ is at most the probability that there exists a nonzero $x$ for which $\abs{\braket{\psi|W_x|\psi}} \ge \sqrt[4]{\frac{1}{3}}$. By \cref{lem:all_weyl_concentration}, this probability is at most $2^{2n + 1} \exp \left(-\frac{2^n}{36 \sqrt{3} \pi^3} \right)$.
\end{proof}

\subsection{Distinguishing From Haar-Random}

We are now ready to state and analyze our algorithm that, given copies of $\ket\psi$, efficiently distinguishes whether $\ket\psi$ is (i) Haar-random or (ii) a state prepared by a $(n/2-1)$-doped Clifford circuit, promised that one of these is the case. 

While the analysis is not so trivial, the algorithm itself is straightforward: Bell difference sample $O(n)$ times, and, with high probability, we will have a set of Weyl operators that span $\F_2^{2n}$ when $\ket\psi$ is Haar-random. 
On the other hand, if $\ket\psi$ is the output of an $t$-doped Clifford circuit, for $t < n/2$, this can never happen because $q_\psi$ is supported on a subspace of dimension at most $n + 2t$ (which we proved in \cref{cor:q-support-dimension}). 

\begin{algorithm}[ht]
\SetKwInOut{Promise}{Promise}
\caption{Distinguishing output of an $(n/2-1)$-doped Clifford circuit from Haar-Random}
\label{alg:distinguisher}
\KwInput{$24n + 18 \log(2/\delta)$ copies of $\ket\psi$}
\Promise{$\ket\psi$ is Haar-random or the output of an $(n/2-1)$-doped Clifford circuit}
\KwOutput{$0$ if $\ket\psi$ is Haar-random and $1$ otherwise, with probability at least $1-\delta$}
Let $m = 6n + \frac{9}{2}\log(2/\delta)$

Let $T = \{\}$

\RepTimes{$m$}{
Perform Bell difference sampling to obtain $x \in \F_2^{2n}$

Add $x$ to $T$
}

Compute the dimension $d$ of the span of $T$ using Gaussian elimination. 

\Return{$0$ if $d = 2n$ and $1$ otherwise.}
\end{algorithm}

To prove the correctness of \Cref{alg:distinguisher}, we need the following lemma.

\begin{lemma}\label{lem:haar_sampling}
Let $\ket\psi$ be an $n$-qubit Haar-random quantum state and fix $\delta > 0$. Taking $6n + \frac{9}{2}\log(2/\delta)$ samples from $q_\psi$ suffices to sample $2n$ linearly independent elements of $\F_2^{2n}$ with probability at least $1-\delta$ over both the Haar measure and the sampling process. 
\end{lemma}
\begin{proof}
    For samples $x_1, \cdots x_m \in \F_2^{2n}$, let $T_i = \langle x_1, \dots , x_{i-1} \rangle$ be the subspace spanned by the first $i-1$ samples for arbitrary $1 \leq i \leq m$.
    Define the $\{0, 1\}$ indicator random variable 
    \[
    X_i  \coloneqq \begin{cases}
        1 & \text{$x_i \not \in T_i$ or $T_i = \F_2^{2n}$}\\
        0 & \text{otherwise}
    \end{cases}
    \]
    such that we have achieved our goal if and only if $\sum_{i=1}^m X_i \geq 2n$.
    We see that if $T_i = \F_2^{2n}$ then $\E[X_i] = 1 \geq \frac{1}{3}$.
    Otherwise, the probability that $x \not\in T_i$ is $\sum_{x \in \F_2^{2n} \setminus T_i} q_\psi(x)$.
    Let $T_i'\supseteq T_i$ be some arbitrary $2n-1$ dimensional extension of $T_i$, such that \[\sum_{x \in T_i} q_\psi(x) \leq \sum_{x \in T_i'} q_\psi(x).\]
    By \cref{lem:proper-subspace-concentration}, $\sum_{x \in T_i'} q_\psi(x) \leq \frac{2}{3}$ for all $T'$, with overwhelmingly high probability over the Haar measure.
    Let us assume that this has happened.
    Since $\sum_{x \in \F_2^{2n}} q_\psi(x) = 1$, we know that in the scenario where $T_i \neq \F_2^{2n}$ that $\sum_{x \in \F_2^{2n} \setminus T_i} q_\psi(x) \geq \frac{1}{3}$.
    Since both scenarios give expectation over $\frac{1}{3}$, this gives us $\E[X_i \mid X_1, \dots, X_{i-1}] \geq \frac{1}{3}$ for any assignment of $X_1,\dots,X_{i-1}$.
    By Hoeffding's inequality (\cref{fact:hoeffding}), 
    \[
        \Pr\left[\sum_{i=1}^m X_i \leq 2n\right] \leq \Pr\left[\sum_{i=1}^m X_i \leq \E\left[\sum_{i=1}^m X_i\right]  - \frac{m}{3} + 2n\right] \leq e^{-\frac{2}{m}\left(\frac{m}{3} - 2n\right)^2}.\footnote{A cautious reader may object that the random variables here are not independent, which is needed for Hoeffding's inequality. However, because the conditional expectations are at least $1/3$, the sum (first-order) stochastically dominates a sum of independent Bernoulli random variables with means $1/3$. Hence, the left tail probability of the former can only be smaller than the latter. }
    \]
    Writing $m = (6+\gamma)n$, this bound becomes
    \[
    \exp\left(-\frac{2\gamma^2n}{9(6+\gamma)}\right) \leq \exp\left(-\frac{2\gamma n}{9}\right)
    \]
    
    Thus taking $\gamma = \frac{9}{2n}\log(2/\delta)$, we see that after $m = 6n + \frac{9}{2}\log(2/\delta)$ samples, this probability is at most $\delta/2$. By the union bound, the total failure probability over both the Haar measure and the samples is at most 
    \[
    \frac{\delta}{2} + 2^{2n+1}\exp \left(- \dfrac{2^{n}}{36 \sqrt{3}\pi^3}\right)
    \]
    which in turn is at most $\delta$, for reasonable choices of $\delta$.\footnote{Of course, the union bound fails when $\delta$ is doubly exponentially small, as our bound for the error over the Haar measure is $2^{-2^{O(n)}}$. However, in this setting, it is information-theoretically impossible to distinguish a state from Haar-random.} 
\end{proof}

\begin{theorem}
    \Cref{alg:distinguisher} succeeds with probability at least $1-\delta$, and it uses $O(n+\log(1/\delta))$ copies of the input state and $O(n^3+n^2\log(1/\delta))$ time. 
    \label{thm:prs_distinguisher}
\end{theorem}
\begin{proof}
    In the case that $\ket\psi$ is the output of a $t$-doped Clifford circuit, $q_\psi$ is supported on a subspace of dimension at most $n+2t$ by \Cref{cor:q-support-dimension}, so the algorithm always outputs $1$. 
    Therefore, we need only argue that in the case that $\ket \psi$ is Haar-random we see at least $2n$ linearly independent elements of $\F_2^{2n}$. By \cref{lem:haar_sampling}, this happens with probability at least $1-\delta$ over the Haar measure and the sampling process.

    It is clear that the sample complexity is $O(m) = O(n + \log(1/\delta))$. To analyze the time complexity, we note that each sample takes $O(n)$ time, so sampling takes $O(mn)$ time. Gaussian elimination takes $O(mn^2) = O(n^3+n^2\log(1/\delta))$ time and is the dominating term.
\end{proof}

Our distinguishing algorithm immediately implies a lower bound on the number of non-Clifford gates needed to prepare computationally pseudorandom quantum states. 

\begin{corollary}
    \label{cor:prs-linear-lower-bound}
    Any family of $t$-doped Clifford circuits that produces an ensemble of $n$-qubit computationally pseudorandom quantum states must satisfy $t \ge n/2$.
\end{corollary}

\begin{proof}
    Let $\ket{\varphi_k}_{k \in \{0,1\}^\kappa}$ be an ensemble of $n$-qubit pseudorandom states, as in \cref{def:prs}. If each $\ket{\varphi_k}$ were computable by a $t$-doped Clifford circuit for some $t < n/2$, then \cref{thm:prs_distinguisher} implies that \cref{alg:distinguisher} with (say) $\delta = 1/3$ would violate the computational indistinguishability criterion in \cref{def:prs}, because it runs in $\poly(n) = \poly(\kappa)$ time. Hence, we must have $t \ge n/2$.
\end{proof}

Note that this lower bound can be improved by a factor of $2$ in the special case that all of the non-Clifford gates are diagonal (e.g.\ $T$-gates), because of the improved lower bound on stabilizer dimension in \cref{lem:arbitrary-gate-dimension} for this case.

\begin{corollary}
    \label{cor:prs-linear-lower-bound-t-gate}
    Any family of Clifford+$T$ circuits that produces an ensemble of $n$-qubit computationally pseudorandom quantum states must use at least $n$ $T$-gates.
\end{corollary}

\section{Stabilizer State Approximations}
\label{sec:stabilizer-approximations}
We state and analyze our algorithm that, given copies of an $n$-qubit quantum pure state $\ket\psi$ that has stabilizer fidelity at least $\tau$, outputs a succinct description of a stabilizer state that witnesses fidelity at least $\tau - \eps$. 

Our presentation is split into two parts. 
First, in \cref{subsec:makes-progress}, we prove a useful lemma regarding $q_{\psi}$ on $S^* = \weyl(\ket\phi)$, where $\ket{\phi}$ is the stabilizer state that maximizes stabilizer fidelity with $\ket\psi$. 
At a high level, we argue that any sample from $q_{\psi}$ has a good enough chance of \say{making progress} towards learning a complete set of generators for $S^*$. Formally, we prove that the $q_{\psi}$-mass on $S^*$ is not heavily concentrated on proper subspaces of $S^*$, so that when we sample an element of $S^*$, we obtain an element of $S^*$  that is linearly independent of the previous samples with a reasonable probability.
Second, in \cref{subsec:stab-fidelity-algo}, we state our algorithm, prove its correctness, and analyze its sample and time complexities. 

\subsection{Bell Difference Sampling Makes Progress}
\label{subsec:makes-progress} 

The next lemma gives a way to argue that in many of our proofs, we can suppose without loss of generality that $\ket{0^n}$ maximizes stabilizer fidelity.

\begin{lemma}
\label{lem:clifford_invariance}
Given an $n$-qubit stabilizer state $\ket{\psi}$, let $S = \weyl(\ket\psi)$ be its unsigned stabilizer group, and let $T \subseteq S$ be a subspace of dimension $n - t$. Then there exists a Clifford circuit $C$ such that $C\ket{\psi} = \ket{0^n}$, $C(S) = 0^n \times \F_2^n$, and $C(T) = 0^{n+t} \times \F_2^{n-t}$.
\end{lemma}

\begin{proof}
    Because the Clifford group acts transitively on stabilizer states, there exists a Clifford circuit $C$ such that $C\ket{\psi} = \ket{0^n}$. Because $S = \{x \in \F_2^{2n} : \bra{\psi}W_x \ket{\psi} = \pm 1\}$, this $C$ necessarily maps $S$ to $C(S) = \{x \in \F_2^{2n} : \bra{\psi}C^\dagger W_x C\ket{\psi} = \pm 1\} = 0^n \times \F_2^n$. So, it only remains to show that $C$ can be chosen so as to map $T$ to $0^{n+t} \times \F_2^{n-t}$ while preserving these properties. This holds because a CNOT gate between qubits $i$ and $j$ in its action on $\F_2^{2n}$ maps $(0^n, x) \in \F_2^{2n}$ to $(0^n, Mx)$ where $M \in \mathrm{GL}_n(\F_2)$ is an elementary matrix (in particular, a matrix equal to the identity except with the $(i, j)$ entry equal to $1$). Hence, CNOT gates between arbitrary qubits generate all of $\mathrm{GL}_n(\F_2)$. So, we can choose CNOT gates so as to map $T$ to an arbitrary subspace of $0^n \times \F_2^n$ of the same dimension, while preserving $\ket{0^n}$ and $0^n \times \F_2^n$.
\end{proof}

Now, we show that the $p_\psi$-mass on $S^*$ is bounded below by the squared stabilizer fidelity of $\ket\psi$.

\begin{lemma}\label{lem:f_lowerbound}
Given an $n$-qubit state $\ket{\psi}$, let $\ket{\phi}$ be a stabilizer state that maximizes the stabilizer fidelity, and let $S^* = \weyl(\ket{\phi})$. Then
\[\sum_{x \in S^*} p_\psi(x) \geq F_\stabset(\ket{\psi})^2.\]
\end{lemma}
\begin{proof}
    Since $\ket{\phi}$ maximizes the stabilizer fidelity, we can write $F_\stabset(\ket{\psi}) = \abs{\braket{\phi | \psi}}^2$.
    Let $C$ be a Clifford circuit from \cref{lem:clifford_invariance} such that $C \ket{\phi} = \ket{0^n}$ and $C(S^*) = 0^n \times \F_2^n$ (the choice of $T$ is irrelevant).
    Now let $\ket{\psi'} = C \ket{\psi}$.
    Based on this mapping,
    \[
        \sum_{x \in S^*} p_\psi(x) = \sum_{x \in C(S^*)} p_{\psi'}(x) = \sum_{x \in \F_2^{n}} p_{\psi'}(0^n, x).
    \]

    It remains to lower bound $\sum_{x \in \F_2^{n}} p_{\psi'}(0^n, x)$.
    
    \begin{align*}
        \sum_{x \in \F_2^{n}} p_{\psi'}(0^n, x) &= \frac{1}{2^n} \sum_{x \in \F_2^n} \braket{\psi' | W_{0,x} | \psi'}^2\\
        &\ge \frac{1}{4^n} \left(\sum_{x \in \F_2^n} \braket{\psi' | W_{0,x} | \psi'}\right)^2\\
        &= \frac{1}{4^n} \left(2^n \abs{\braket{\psi' | 0^n}}^2\right)^2\\
        &= \abs{\braket{\psi| C^\dagger| 0^n}}^4\\
        &= \abs{\braket{\psi | \phi}}^4 = F_\stabset(\ket{\psi})^2.
    \end{align*}
    Since we know that $\sum_{x \in \F_2^{n}} p_{\psi'}(0^n, x) \geq F_\stabset(\ket{\psi})^2$, this tells us that $\sum_{x \in S^*} p_\psi(x) \geq F_\stabset(\ket{\psi})^2$ as well.
\end{proof}
We can generalize this result to arbitrary subspaces of $S^*$.
\begin{corollary}\label{cor:f_lowerbound_subgroup}
    Given an $n$-qubit state $\ket{\psi}$, let $\ket{\phi}$ be a stabilizer state that maximizes the stabilizer fidelity, and let $S^* = \weyl(\ket{\phi})$. Let $T \subseteq S^*$ be a subspace of $S^*$. Then \[
        \sum_{x \in T} p_\psi(x) \geq \frac{\abs{T}}{2^n} F_\stabset(\ket{\psi})^2.
    \]
\end{corollary}
\begin{proof}
    \begin{align*}
        \sum_{x \in T} p_\psi(x) &= \frac{\abs{T}}{2^n} \sum_{x \in T^\sympcomp} p_\psi(x) & \text{(\Cref{thm:p-mass-identity-subgroups})} \\
        &\geq \frac{\abs{T}}{2^n} \sum_{x \in S^*} p_\psi(x) &  \\
        &= \frac{\abs{T}}{2^n} F_\stabset(\ket{\psi})^2,  & \text{(\Cref{lem:f_lowerbound})}
    \end{align*}
    where we have used the fact that  $S^* \subseteq T^\perp$.
\end{proof}

Now we show a series of anticoncentration lemmas on proper subspaces of $S^*$. For these next lemmas, we will find it more convenient to assume without loss of generality (because of \cref{lem:clifford_invariance}) that the state maximizing fidelity is $\ket{0^n}$, which conceptually simplifies the computations.

\begin{figure}
    \centering
    \includegraphics[width=3in]{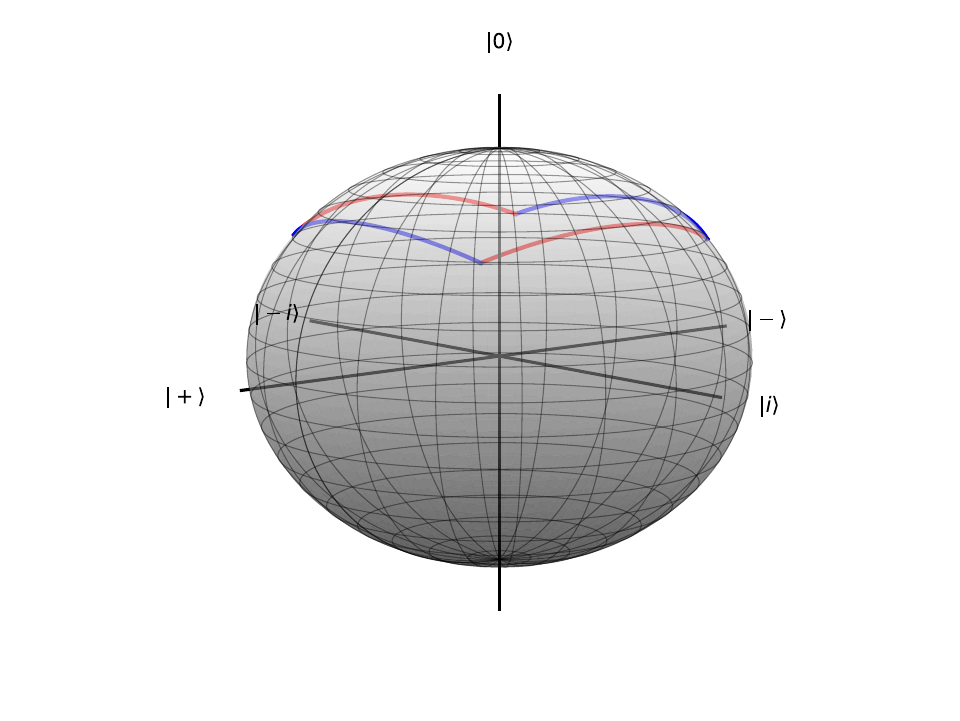}
    \includegraphics[width=3in]{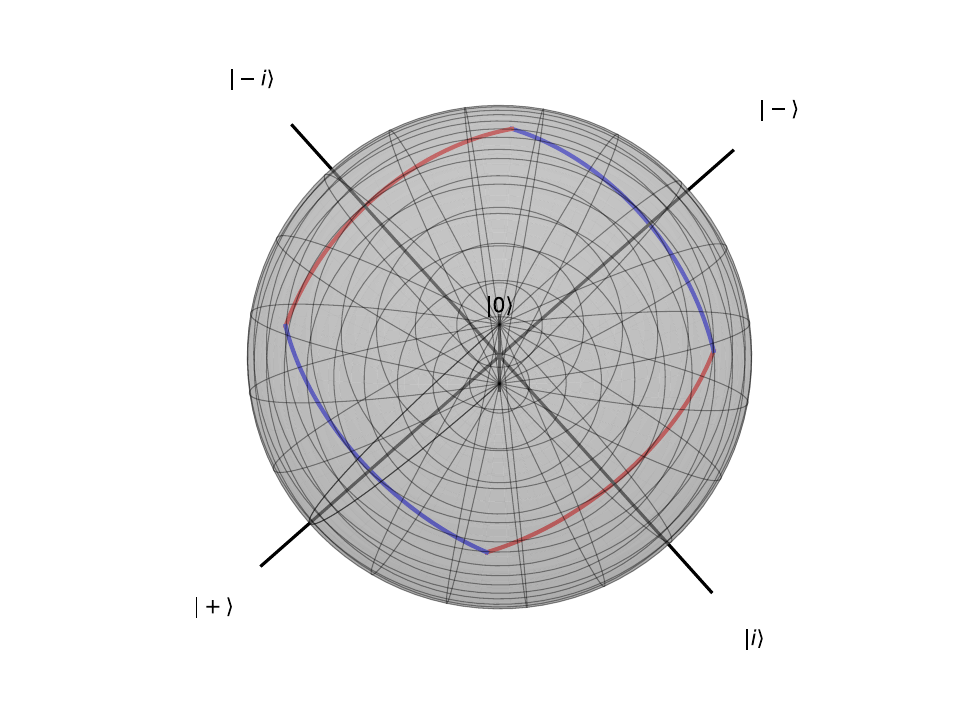}
    \caption{
    An illustration of the argument in the proof of \Cref{lem:c_mass-S-minus-T}.
    Consider the (possibly unnormalized) state $\alpha_0\ket{0^n}+\alpha_1\ket{10^{n-1}}$. We can visualize the first qubit of this state on the Bloch sphere.
    The surface enclosed by the red and blue curve is exactly the set of points on the sphere for which $\ket{0}$ is the closest stabilizer state. By our assumption that the stabilizer fidelity of $\ket\psi$ is maximized by $\ket{0^n}$, $\alpha_0 \ket{0} + \alpha_1 \ket{1}$ must lie on this surface, up to normalization. The corners of this surface (the intersection of a blue curve with a red curve) represent a choice of $\alpha_0$ and $\alpha_1$ that minimizes $\alpha_0$.}
    \label{fig:fidelity-proof-diagram}
\end{figure}

\begin{lemma}\label{lem:c_mass-S-minus-T}
Let $\ket{\psi}$ be an $n$-qubit state. Suppose the fidelity $\abs{\braket{\psi|\phi}}^2$ is maximized by $\ket{\phi} = \ket{0^n}$ over stabilizer states $\ket{\phi}$. Let $S^* = 0^n \times \F_2^n = \weyl(\ket{0^n})$, and let $T = 0^{n + 1} \times \F_2^{n - 1}$ be a maximal subspace of $S^*$.
Then

    \[\sum_{x \in S^* \setminus T} c_\psi(x) \geq 2^{\frac{n}{2}-1} \left(\sqrt{3}-1\right) F_\stabset(\ket{\psi}).
    \]
\end{lemma}
\begin{proof}
We can express the sum as
\begin{align*}
    \sum_{x \in S^* \setminus T} c_\psi(x)
    &= \frac{1}{\sqrt{2^n}} \sum_{x \in 1 \times \F_2^{n-1} } \Tr\left[\ketbra{\psi}{\psi} W_{0^n, x}\right]\\
    &= 2^{\frac{n}{2}-1} \Tr\left[\ketbra{\psi}{\psi} \left(Z \otimes \ketbra{0}{0}^{n-1} \right) \right]\\
    &= 2^{\frac{n}{2}-1} \left(\abs{\alpha_0}^2 - \abs{\alpha_1}^2\right),
\end{align*}
where $\alpha_0$ is the amplitude of $\ket{\psi}$ on $\ket{0^n}$ and $\alpha_1$ is its amplitude on $\ket{10^{n-1}}$. Note that $|\alpha_0|^2 = F_\stabset(\ket{\psi})$, by assumption. Thus, we need to show that $\abs{\alpha_1}$ cannot be too big compared to $\abs{\alpha_0}$, or else it would contradict the maximality of $\abs{\braket{\psi|\phi}}^2$ at $\ket{\phi} = \ket{0^n}$. We give a visual proof of this fact in \cref{fig:fidelity-proof-diagram}, along with an algebraic proof below.

Choose the global phase on $\ket{\psi}$ to assume without loss of generality that $\alpha_0$ is positive-real and $\alpha_1 = \abs{\alpha_1}e^{i \theta}$. We may write:
\begin{align*}
    \abs{\left(\bra{+} \otimes \bra{0^{n-1}} \right)\ket{\psi}}^2 = \frac{1}{2} \abs{\alpha_0 + \alpha_1}^2 &= \frac{1}{2}\left(\left( \alpha_0 + \abs{\alpha_1} \cos \theta\right)^2 + \abs{\alpha_1}^2 \sin^2 \theta \right)\\
    \abs{\left(\bra{-} \otimes \bra{0^{n-1}}\right)\ket{\psi}}^2 = \frac{1}{2} \abs{\alpha_0 - \alpha_1}^2 &= \frac{1}{2}\left(\left( \alpha_0 - \abs{\alpha_1} \cos \theta\right)^2 + \abs{\alpha_1}^2 \sin^2 \theta \right)\\
    \abs{\left(\bra{i} \otimes \bra{0^{n-1}}\right) \ket{\psi}}^2 = \frac{1}{2} \abs{\alpha_0 + i\alpha_1}^2 &=\frac{1}{2}\left(\left( \alpha_0 - \abs{\alpha_1} \sin \theta\right)^2 + \abs{\alpha_1}^2 \cos^2 \theta \right)\\
    \abs{\left(\bra{-i} \otimes \bra{0^{n-1}}\right)\ket{\psi}}^2 = \frac{1}{2} \abs{\alpha_0 - i\alpha_1}^2 &=\frac{1}{2}\left(\left( \alpha_0 + \abs{\alpha_1} \sin \theta\right)^2 + \abs{\alpha_1}^2 \cos^2 \theta \right).
\end{align*}
All of these values need to be less than $\abs{\alpha_0}^2$, as otherwise $\ket{\psi}$ would have larger fidelity with one of the above states.
Due to symmetry of both $\sin$ and $\cos$, we will only consider $\theta \in [0, \frac{\pi}{2}]$ such that the only relevant equations to consider are the first and last.
This allows us to write the largest of the above inner products as
\[
    \frac{1}{2}\left(\alpha_0^2 + \abs{\alpha_1}^2 + 2\alpha_0 \abs{\alpha_1} \cdot \max\left(\cos \theta, \sin \theta\right)\right),
\]
which is minimized for $\theta = \pi/4$.
Plugging that back in and comparing to $\alpha_0^2$ leads to
\[
    \alpha_0^2 \geq \frac{1}{2}\left(\alpha_0^2 + \abs{\alpha_1}^2 + \sqrt{2}\alpha_0 \abs{\alpha_1} \right),
\]
and solving for the maximum $\abs{\alpha_1}$ gives $\abs{\alpha_1} \leq \left(2-\sqrt{3}\right)^{1/2}\abs{\alpha_0}$. Hence, $\abs{\alpha_0}^2 - \abs{\alpha_1}^2 \geq 1-(2-\sqrt{3}) \alpha_0^2 = \left(\sqrt{3} - 1\right)F_\stabset(\ket{\psi})$.
Therefore,
\[
\sum_{x \in S^* \setminus T} c_\psi(x) = 2^{\frac{n}{2}-1} \left(\abs{\alpha_0}^2 - \abs{\alpha_1}^2\right) \geq 2^{\frac{n}{2}-1} \left(\sqrt{3}-1\right) F_\stabset(\ket{\psi}). \qedhere
\]

\end{proof}

\begin{lemma}\label{lem:q-S-minus-t-zeros}
Let $\ket{\psi}$ be an $n$-qubit state. Suppose the fidelity $\abs{\braket{\psi|\phi}}^2$ is maximized by $\ket{\phi} = \ket{0^n}$ over stabilizer states $\ket{\phi}$. Let $S^* = 0^n \times \F_2^n = \weyl(\ket{0^n})$, and let $T = 0^{n + 1} \times \F_2^{n - 1}$ be a maximal subspace. Then
    \[\sum_{x \in S^* \setminus T}q_\psi(x) \geq \frac{2-\sqrt{3}}{2} F_\stabset(\ket{\psi})^4.\]
\end{lemma}
\begin{proof}
We can write
\begin{align*}
    \sum_{x \in S^* \setminus T}q_\psi(x) &= \sum_{x \in S^* \setminus T}\sum_{t \in \F_2^{2n}} p_\psi(t) p_\psi(x+t) \\ 
    &\geq \sum_{t \in T} p_\psi(t)  \sum_{x \in S^* \setminus T} p_\psi(x+t) \\
    &= \left(\sum_{t \in T} p_\psi(t)\right) \left(\sum_{x' \in S^* \setminus T} p_\psi(x')\right) && (t + S^* \setminus T = S^* \setminus T)\\
    & \geq \frac{1}{2^{n-1}} \left(\sum_{t \in T} p_\psi(t)\right) \left(\sum_{x' \in S^* \setminus T} \abs{c_\psi(x')}\right)^2 && \textrm{(Cauchy-Schwarz)}\\
    &\geq \frac{1}{2^{n-1}}\left(\sum_{t \in T} p_\psi(t)\right) \left(\sum_{x' \in S^* \setminus T} c_\psi(x')\right)^2.
\end{align*}
Now apply \cref{cor:f_lowerbound_subgroup} and \cref{lem:c_mass-S-minus-T} respectively and we get
\[
    \sum_{x \in S^* \setminus T}q_\psi(x) \geq \frac{1}{2^{n-1}} \cdot \frac{\abs{T}}{2^n} F_\stabset(\ket{\psi})^2 \cdot \left(2^{\frac{n}{2}-1} \left(\sqrt{3}-1\right) F_\stabset(\ket{\psi})\right)^2 = \frac{2-\sqrt{3}}{2} F_\stabset(\ket{\psi})^4
\]
as desired.
\end{proof}

\begin{lemma}\label{lem:q-S-minus-t}
Given an $n$-qubit state $\ket{\psi}$, let $\ket{\phi}$ be a stabilizer state that maximizes the stabilizer fidelity, and let $S^* = \weyl(\ket{\phi})$. Let $T \subset S^*$ be a proper subspace of $S^*$. Then
    \[\sum_{x \in S^* \setminus T}q_\psi(x) \geq \frac{2-\sqrt{3}}{2} F_\stabset(\ket{\psi})^4.\]
\end{lemma}

\begin{proof}
    Use \cref{lem:clifford_invariance} to choose a Clifford circuit such that $C\ket{\phi} = \ket{0^n}$, $C(S^*) = 0^n \times \F_2^n$, and $C(T) \subseteq 0^{n+1} \times \F_2^{n-1}$. Let $\ket{\psi'} = C\ket{\psi}$.
    Then:
    \[
    \sum_{x \in S^* \setminus T} q_\psi(x) = \sum_{x \in C(S^* \setminus T)} q_{\psi'}(x) \ge \sum_{x \in 0^n \times 1 \times \F_2^{n-1}} q_{\psi'}(x).
    \]
    By \cref{lem:q-S-minus-t-zeros}, this sum is lower bounded by
    \[
    \sum_{x \in 0^n \times 1 \times \F_2^{n-1}} q_{\psi'}(x) \ge \frac{2 - \sqrt{3}}{2} F_\stabset(\ket{\psi'})^4 = \frac{2 - \sqrt{3}}{2} F_\stabset(\ket{\psi})^4.\qedhere
    \]
\end{proof}

\subsection{The Algorithm}\label{subsec:stab-fidelity-algo}
Our algorithm for stabilizer state approximations uses the powerful classical shadows framework \cite{HKP20-classical-shadows} to improve its sample complexity.

\begin{theorem}[Classical shadows algorithm \cite{HKP20-classical-shadows}]
\label{thm:classical_shadows}
Let $\rho$ be an unknown $n$-qubit mixed state. Then there exists a quantum algorithm that first performs $m_{\mathsf{shadow}} = O(\log(K / \delta) / \eps^2)$ random Clifford measurements on independent copies of $\rho$. Then, later given $K$ different observables $O_1, O_2, \ldots, O_K$ in an online fashion, where each $O_i$ is a rank-$1$ projector, the algorithm uses the measurement results to output estimates $\hat{o}_1, \ldots, \hat{o}_K$, such that with probability at least $1 - \delta$, for every $i \in [K]$, $|\hat{o}_i - \tr(O_i \rho)| \le \eps$. Moreover, if $O_i$ is a projector onto a stabilizer state, then each $\hat{o}_i$ can be computed from the measurement results by a classical algorithm that takes time $O(n^2 m_{\mathsf{shadow}})$.
\end{theorem}

For the ``moreover'' part of \cref{thm:classical_shadows}, see the remarks on Page 1053 of \cite{HKP20-classical-shadows}.\footnote{This is page 1053 of \emph{Nature Physics} Volume 16. Alternatively, see page 5 of the arXiv version.} 

We also require an algorithm, due to \cite{tomita2006worst}, for computing all of the maximal cliques in a graph.

\begin{theorem}[Computing maximal cliques \cite{tomita2006worst}]
    \label{thm:maximal_clique_alg}
    Given an undirected graph $G$ with $n$ vertices, there is a classical algorithm that outputs a list of all of the maximal cliques in $G$ in time $O(3^{n/3})$.
\end{theorem}

Note that this implies that the number of maximal cliques is at most $O(3^{n/3})$.

We are now ready to describe the stabilizer state approximation algorithm. At a high level, it uses Bell difference sampling to obtain a list of candidate Lagrangian subspaces generated by the sampled Weyl operators. Then, it iterates through the candidate groups to find the stabilizer state with largest fidelity, using classical shadows to perform the estimation.

\begin{algorithm}
\caption{Stabilizer State Approximation}
\label{alg:fidelity-estimation}
\SetKwInOut{Promise}{Promise}
\KwInput{$m_{\mathsf{shadow}} + 4m_{\mathsf{clique}}$ copies of $\ket\psi$}
\Promise{$\ket\psi$ has stabilizer fidelity at least $\tau$}
\KwOutput{A stabilizer state $\ket\phi$ such that $\abs{\braket{\phi|\psi}}^2 \geq F_\stabset(\ket{\psi}) - \eps$ with probability at least $1 - \delta$}

Initialize an empty graph $G$

\RepTimes{$m_{\mathsf{clique}}$}{
Using $4$ copies of $\ket{\psi}$, perform Bell difference sampling to obtain $x \in \F_2^{2n}$

Add a vertex for $x$ in $G$ and connect it to all vertices $y$ in $G$ such that $[x,y] = 0$.
}

\RepTimes{$m_{\mathsf{shadow}}$}{
Choose a random Clifford circuit $U$

Measure $U\ket{\psi}$ in the computational basis and store the result
}

\ForEach{maximal clique $(v_1, \ldots, v_k) \in G$ computed using \cref{thm:maximal_clique_alg}}{

Compute $S \coloneqq \langle v_1, \ldots, v_k \rangle$ via Gaussian elimination

\If{$|S| = 2^n$}{

\ForEach{stabilizer state $\ket{\phi}$ with $\weyl(\ket \phi) = S$}{

Let $\hat{o}_{\phi}$ be the estimator of $\abs{\braket{\psi|\phi}}^2$ computed using the algorithm in \cref{thm:classical_shadows}
}

}
}
\Return{
whichever $\ket{\phi}$ maximizes $\hat{o}_\phi$
}

\end{algorithm}

We first argue that with high probability, one of the maximal cliques generates the Lagrangian subspace corresponding to a state that maximizes stabilizer fidelity.

\begin{lemma}
\label{lem:fidelity_finds_samples}
Given an $n$-qubit state $\ket{\psi}$, let $\ket{\phi}$ be a stabilizer state that maximizes the stabilizer fidelity, and let $S^* = \weyl(\ket{\phi})$. Suppose $\abs{\braket{\phi|\psi}}^2 \ge \tau$.
Then choosing $m_\mathsf{clique} \ge \frac{8 + 4\sqrt{3}}{\tau^4}(n + \log(1/\delta))$ is sufficient to guarantee that with probability at least $1 - \delta$, the Bell difference sampling step of \Cref{alg:fidelity-estimation} samples a complete set of generators for $S^*$.
\end{lemma}

\begin{proof}
For notational simplicity, write $m_\mathsf{clique} = m$.
Let $x_1,...,x_{m} \in \F_2^{2n}$ be the results of the Bell difference sampling subprocedure in \Cref{alg:fidelity-estimation}.
Let $T_{i}$ be the subspace of $S^*$ spanned by all elements in $\{x_1,\ldots,x_i\} \cap S^*$. Define the indicator random variable $X_i$ as

\[
X_i = \begin{cases}
1 & x_i \in S^* \setminus T_i \textrm{ or } T_i = S^* \\
0 & \textrm{otherwise}.
\end{cases}
\]

Informally, $X_i = 1$ indicates a step at which the algorithm has made progress towards sampling a complete set of generators for $S^*$. Since $S^*$ has dimension $n$, we need to show that with high probability, $\sum_{i=1}^{m} X_i \geq n$, as this guarantees that $T_{m} = S^*$.

From \Cref{lem:q-S-minus-t} we have that any assignment of $X_1,\ldots,X_{i-1}$, $\E[X_i | X_1, \ldots, X_{i-1}] \geq c$, where $c = \frac{2 - \sqrt{3}}{2}\tau^4$, because $x_i$ is sampled with probability $q_\psi(x_i)$.
Let $\gamma = 1 - \frac{n}{cm}$.
Then, by a Chernoff bound (\cref{fact:chernoff}), we have
\begin{align}
\Pr\left[\sum_{i=1}^{m} X_i \leq n \right]
&= \Pr\left[\sum_{i=1}^{m} X_i \leq (1 - \gamma)cm \right]\nonumber\\
&\leq \exp\left(- \gamma^2 cm / 2\right)\nonumber\\
&= \exp\left(- \left(1 - \frac{2n}{cm} + \frac{n^2}{c^2m^2}\right) cm / 2\right)\nonumber\\
&\le \exp\left(- \left(1 - \frac{2n}{cm}\right) cm / 2\right)\nonumber\\
&= \exp\left(n - cm/2\right)\label{eq:chernoff_fidelity_estimate}
\end{align}
Hence, choosing
\[
m \ge \frac{2}{c}(n + \log(1/\delta)) = \frac{8 + 4\sqrt{3}}{\tau^4}(n + \log(1/\delta))
\]
suffices to guarantee that \cref{eq:chernoff_fidelity_estimate} is at most $\delta$.\qedhere
\end{proof}

Now we have everything needed to prove the correctness of \cref{alg:fidelity-estimation}.

\begin{theorem}\label{thm:fidelity_learning}
    Let $\ket{\psi}$ be an $n$-qubit state with $F_\stabset(\ket{\psi}) \ge \tau$. Then choosing 
    \begin{align*}
    m_\mathsf{clique} = \frac{8 + 4\sqrt{3}}{\tau^4}(n + \log(2/\delta)) && m_\mathsf{shadow} = O\left(\frac{n + \log(1/\delta)}{\eps^2 \tau^4}\right)
    \end{align*}
    suffices to guarantee that with probability at least $1 - \delta$, \cref{alg:fidelity-estimation} outputs a state $\ket{\phi}$ satisfying $\abs{\braket{\phi|\psi}}^2 \ge F_\stabset(\ket{\psi}) - \eps$ and it uses $O\left(\frac{n + \log(1/\delta)}{\eps^2\tau^4}\right)$ samples and $\exp\left(O\left(\frac{n + \log(1/\delta)}{\tau^4}\right)\right)\cdot \frac{1}{\eps^2}$ time. 
\end{theorem}

\begin{proof}
     Choose the failure probability in \cref{lem:fidelity_finds_samples} to be at most $\delta / 2$. Choose the parameters in \cref{thm:classical_shadows} so that the additive error in the estimates is $\eps / 2$ and the failure probability is at most $\delta / 2$; this requires choosing $K = 2^n \cdot O\left(3^{m_\mathsf{clique} / 3}\right)$ and thus $m_\mathsf{shadow} = O(\log(K/\delta)/\eps^2) = O((n + m_\mathsf{clique} + \log(1/\delta)) / \eps^2) = O(m_\mathsf{clique} / \eps^2)$.
     
     We assume henceforth that both \cref{thm:classical_shadows} and \cref{lem:fidelity_finds_samples} do not fail, which happens with probability at least $1 - \delta$ over the samples.
    
    Letting $\ket{\varphi}$ be the state maximizing stabilizer fidelity and $S^* = \weyl(\ket{\varphi})$, \cref{lem:fidelity_finds_samples} guarantees that the algorithm samples a complete set of generators for $S^*$. These generators are necessarily contained in some maximal clique of $G$ because they all commute, and moreover, the subspace spanned by this clique must equal $S^*$ because $S^*$ equals its symplectic complement (so the maximal clique cannot contain any elements not in $S^*$).

    By \cref{thm:classical_shadows}, the estimate $\hat{o}_\varphi$ is at least $F_\stabset(\ket{\psi}) - \eps / 2$, so $\max_\phi \hat{o}_\phi \ge F_\stabset(\ket{\psi}) - \eps / 2$. Thus, the state $\ket{\phi}$ that maximizes the estimate $\hat{o}_\phi$ (and is output by the algorithm) has $\abs{\braket{\phi|\psi}}^2 \ge \hat{o}_\phi - \eps / 2 \ge \hat{o}_\varphi - \eps / 2 \ge F_\stabset(\ket{\psi}) - \eps$.

Finally, we briefly comment on the sample and time complexities of \cref{alg:fidelity-estimation}.
The sample complexity of the algorithm is 
\[
m_{\mathsf{shadow}} + 4m_{\mathsf{clique}} = O\left(\frac{n + \log(1/\delta)}{\eps^2\tau^4}\right).
\]
 The runtime is dominated by iterating through all of the maximal cliques, iterating through all of the stabilizer states $\ket{\phi}$ such that $\weyl(\ket\phi) = S$, and computing $\hat{o}_\phi$. There are at most $O\left(3^{m_\mathsf{clique}/3}\right)$ maximal cliques, by \cref{thm:maximal_clique_alg}. There are exactly $2^n$ stabilizer states in each basis. \cref{thm:classical_shadows} then guarantees that computing each $\hat{o}_\phi$ from the classical shadows takes time $O(n^2 m_\mathsf{shadow})$. Thus the overall time complexity is at most
\[
O\left(3^{m_\mathsf{clique}/3} \cdot 2^n \cdot n^2 m_\mathsf{shadow}\right).
\]
Plugging in the bounds on $m_\mathsf{clique}$ and $m_\mathsf{shadow}$ gives
\[
\exp\left( O\left(\frac{n + \log(1/\delta)}{\tau^4}\right)\right) \cdot \frac{n^3 + n^2\log(1/\delta)}{\eps^2 \tau^4},
\]
which further simplifies to
\[
\exp\left( O\left(\frac{n + \log(1/\delta)}{\tau^4}\right)\right) \cdot \frac{1}{\eps^2}
\]
by absorbing the rightmost term into the big-$O$ in the exponent.
\end{proof}

\section{Bounded-Distance Stabilizer Approximation}
We present an efficient algorithm that, given copies of a quantum state $\ket{\psi}$ that has fidelity at least $\cos^2\left(\pi/8\right) + \gamma$ with a stabilizer state $\ket{\phi}$, outputs $\ket{\phi}$ with high probability. We note that this algorithm solves the same task as \cref{alg:fidelity-estimation}, but with the difference that it only works for $\tau > \cos^2(\pi/8)$.

We start by bounding the squared expectation of Weyl operators in the unsigned stabilizer group of $\ket{\phi}$.

\begin{proposition}\label{prop:stab-group-expectation-is-large}
Let $\ket{\psi}$ be an $n$-qubit quantum state that has fidelity $\tau$ with a stabilizer state $\ket{\phi}$, where $\tau \geq 1/2$. If $x \in \weyl(\ket{\phi})$, then 
$\braket{\psi|W_x|\psi}^2 \ge (2\tau-1)^2$.
\end{proposition}
\begin{proof}
We write the input state as $\ket{\psi} = \sqrt{\tau}\ket{\phi} + \sqrt{1 - \tau}\ket{\phi^\perp}$, where $\ket{\phi}$ is a stabilizer state and $\ket{\phi^\perp}$ is the part orthogonal to $\ket{\phi}$. Let $P = \pm W_x$ be the signing of $W_x$ for which $\braket{\phi|P|\phi} = 1$.
Then
\begin{align*}
    \braket{\psi|P|\psi} 
    &= \bra{\psi} \left(\sqrt{\tau}\ket{\phi} + \sqrt{1-\tau}P\ket{\phi^\perp} \right) \\
    &= \sqrt{\tau}\braket{\psi|\phi} + \sqrt{1-\tau}\braket{\psi|P|\phi^\perp} \\
    &= \tau + \sqrt{1-\tau}\braket{\psi|P|\phi^\perp} \\
    &= \tau + \sqrt{1-\tau}\left(\sqrt{\tau}\braket{\phi | P | \phi^\perp} + \sqrt{1-\tau}\braket{\phi^\perp | P | \phi^\perp}\right) \\
    &= \tau + (1-\tau)\braket{\phi^\perp|P|\phi^\perp}.
\end{align*}
Since $\braket{\phi^\perp|P|\phi^\perp} \in [-1, 1]$, we conclude that $\braket{\psi|P|\psi} \ge 2\tau-1$, and hence (using $\tau \ge 1/2$) that $\braket{\psi|W_x|\psi}^2 = \braket{\psi|P|\psi}^2 \ge (2\tau-1)^2$.
\end{proof}

If an operator is not in $\ket{\phi}$'s unsigned stabilizer group, then it must anticommute with at least half of the Pauli operators in that group. 
The uncertainty principle states that the expectation of these operators must be small, since the expectation of the operators in $\weyl(\ket{\phi})$ is large.
To show this formally, we use the Schr\"odinger uncertainty relation.

\begin{fact}[Schr\"odinger uncertainty relation \cite{schrodinger1930uncertainty, angelow2008heisenberg}]\label{fact:uncertainty}
        For a quantum state $\ket{\psi}$ and observables $A$ and $B$,
        \[
            \left(\braket{\psi | A^2 | \psi} - \braket{\psi | A | \psi}^2\right)\left(\braket{\psi | B^2 | \psi} - \braket{\psi | B | \psi}^2\right) \geq \left|\frac{1}{2}\braket{\psi |\left(AB + BA\right) | \psi} - \braket{\psi | A |\psi}\braket{\psi | B | \psi}\right|^2.
        \]
\end{fact}

\begin{proposition}\label{prop:upper-bound}
Let $\ket{\psi}$ be an $n$-qubit quantum state that has fidelity $\tau$ with a stabilizer state $\ket{\phi}$, where $\tau \ge 1/2$. 
If $y \notin \weyl(\ket{\phi})$, then 
\[
\braket{\psi|W_y|\psi}^2 \leq 4 \tau(1-\tau).
\]
\end{proposition}
\begin{proof}
Let $W_x$ be a Pauli operator that anticommutes with $W_y$ and for which $x \in \weyl(\ket{\phi})$. Note that such an operator must exist, for if it didn't, then $W_y$ would commute with every stabilizer of $\ket{\phi}$, and therefore we would have $y \in \weyl(\ket{\phi})^\sympcomp = \weyl(\ket{\phi})$, which contradicts the supposition $y \notin \weyl(\ket{\phi})$.
\cref{fact:uncertainty} simplifies to the following:
\[
 \left(1 - \braket{\psi | W_x | \psi}^2\right)\left(1 - \braket{\psi | W_y | \psi}^2\right) \geq \braket{\psi | W_x |\psi}^2\braket{\psi | W_y | \psi}^2,
\]
where we use the fact that all Pauli operators are Hermitian and unitary and that $W_x$ and $W_y$ anticommute. Then
\begin{align*}
&\left(1 - \braket{\psi | W_x | \psi}^2\right)\left(1 - \braket{\psi | W_y | \psi}^2\right) \geq \braket{\psi | W_x |\psi}^2\braket{\psi | W_y | \psi}^2  \\
&\iff 1 - \braket{\psi | W_x | \psi}^2 - \braket{\psi | W_y | \psi}^2 + \braket{\psi | W_x | \psi}^2  \braket{\psi | W_y | \psi}^2  \geq \braket{\psi | W_x |\psi}^2\braket{\psi | W_y | \psi}^2  \\
&\iff 1 - \braket{\psi | W_x | \psi}^2 - \braket{\psi | W_y | \psi}^2 \geq 0 \\
&\iff 1 - \braket{\psi | W_x | \psi}^2   \geq \braket{\psi | W_y | \psi}^2 \\
&\implies 1 -  (2\tau - 1)^2   \geq \braket{\psi | W_y | \psi}^2\\
&\iff 4 \tau (1-\tau) \geq \braket{\psi | W_y | \psi}^2. 
\end{align*}
In the second-to-last step, we used \cref{prop:stab-group-expectation-is-large}.
\end{proof}

\cref{prop:stab-group-expectation-is-large} and \cref{prop:upper-bound} 
suggest that we can determine whether a given Pauli operator is in the unsigned stabilizer group only from its squared expectation as long as, for all $y \notin \weyl(\ket\phi)$ and for all $x \in \weyl(\ket\phi)$,
\[
\braket{\psi|W_y|\psi}^2 < \braket{\psi|W_x|\psi}^2,
\]
which happens only when $4 \tau(1- \tau) < (2\tau - 1)^2 \iff  \cos^2(\pi/8) < \tau$. 
However, we must also take into account the fact that we cannot know the squared expectations exactly. 
Rather, we can only recover them to some $\pm O(\gamma)$ accuracy, which in turn implies that $\tau$ must be at least $\cos^2(\pi/8) + \gamma$ for some $\gamma > 0$. 
We formalize this in the following corollary.

\begin{corollary}
\label{prop:always-larger}
Let $\ket\psi$ be an $n$-qubit quantum state that has fidelity $\cos^2(\pi/8) + \gamma$ with a stabilizer state $\ket{\phi}$ for some $\gamma > 0$. 
Then for all $y \notin \weyl(\ket\phi)$ and all $x \in \weyl(\ket\phi)$,
\[
\braket{\psi|W_x|\psi}^2 \ge \frac{1}{2} + 4 \gamma^2 + 2\sqrt{2}\gamma > \frac{1}{2} + 2\sqrt{2}\gamma,
\]
and 
\[
\braket{\psi|W_y|\psi}^2 \le \frac{1}{2} - 4 \gamma^2 - 2\sqrt{2}\gamma < \frac{1}{2} - 2\sqrt{2}\gamma.
\]
\end{corollary}
\begin{proof}
Plugging $\tau \ge \cos^2(\pi/8) + \gamma$ into \cref{prop:stab-group-expectation-is-large} and \cref{prop:upper-bound} implies the result. 
\end{proof}

A noteworthy consequence of \Cref{prop:always-larger} is that the state $\ket{\phi}$ must be unique:

\begin{corollary}
\label{cor:uniqueness}
    If $\ket{\psi}$ has fidelity at least $\cos^2(\pi/8) + \gamma$ with a stabilizer state $\ket{\phi}$ for some $\gamma > 0$, then $\ket{\phi}$ must be unique. 
\end{corollary}

\begin{proof}
    Suppose towards a contradiction that $\ket{\psi}$ has stabilizer fidelity at least $\cos^2(\pi/8) + \gamma$ with two different stabilizer states $\ket{\phi_1} \neq \ket{\phi_2}$. If $\weyl(\ket{\phi_1}) = \weyl(\ket{\phi_2})$, then $\ket{\phi_1}$ and $\ket{\phi_2}$ must be orthogonal. But then having $\abs{\braket{\psi|\phi_1}}^2 > \cos^2(\pi/8)$ and $\abs{\braket{\psi|\phi_2}}^2 > \cos^2(\pi/8)$ violates normalization of $\ket{\psi}$. In the complementary case, if $\weyl(\ket{\phi_1}) \neq \weyl(\ket{\phi_2})$, then there exists an $x \in \weyl(\ket{\phi_1}) \setminus \weyl(\ket{\phi_2})$. But then \Cref{prop:always-larger} implies both $\braket{\psi|W_x|\psi}^2 > \frac{1}{2}$ and $\braket{\psi|W_x|\psi}^2 < \frac{1}{2}$, a contradiction.
\end{proof}

Observe that the threshold $\cos^2(\pi/8)$ in \Cref{cor:uniqueness} is tight, because $\cos(\pi/8)\ket{0} + \sin(\pi/8)\ket{1}$ has fidelity $\cos^2(\pi/8)$ with both $\ket{0}$ and $\ket{+} = \frac{\ket{0} + \ket{1}}{\sqrt{2}}$.

\subsection{The Algorithm}
We now state and analyze our algorithm.
\cref{prop:always-larger}---which is the starting point of our algorithm---implies that, based only on the squared expectation of a Weyl operator, we can decide whether nor not it is in $\weyl(\ket\phi)$, where $\ket{\phi}$ is the stabilizer state that has fidelity at least $\cos^2(\pi/8)+\gamma$ with the input state $\ket{\psi}.$
At a high level, there are two missing pieces to complete our algorithm.

First, we need to find a polynomial-size list of Weyl operators that is guaranteed to contain a list of generators of $\weyl(\ket\phi)$. By \cref{lem:fidelity_finds_samples}, we can achieve this by Bell difference sampling repeatedly from $\weyl(\ket\phi)$.
Second, we must estimate the squared expectations of the Weyl operators we sample. One way to do so is by na\"ively measuring each Weyl operator repeatedly, one at a time. However, an algorithm due to Huang, Kueng, and Preskill \cite{huang2021information} achieves a better runtime, letting us estimate the squared expectation value of many Weyl operators with only a logarithmic sample complexity. 

\begin{restatable}[{\cite[Proof of Theorem 4]{huang2021information}}]{theorem}{hkpthmrestate}\label{thm:weyl-estimation}
Given any $m$ Weyl operators $P_{1}, \ldots, P_{m}$ and copies of an unknown pure state $\ket \psi$, there is an algorithm that estimates $\braket{\psi|P_{i}|\psi}^2$ to $\pm \eps$ accuracy with probability at least $1 - \delta$ by performing Bell measurements on 
\[
\frac{4 \log(2m/\delta)}{\eps^2}
\]
copies of the unknown state $\ket\psi$. The time it takes is $O\left(n \cdot m \cdot \frac{\log\left(m/\delta \right)}{\eps^2}\right)$.
\end{restatable}

We also provide a simple proof of \cref{thm:weyl-estimation} in \cref{appendix:hkp}.

Putting the pieces together, our algorithm works as follows.

\begin{algorithm}[H]
\SetKwInOut{Promise}{Promise}
\caption{Bounded-Distance Stabilizer Approximation}\label{alg:main-alg}
\KwInput{$O\left(n + \frac{\log(n/\delta)}{\gamma^2}\right)$ copies of $\ket\psi$} 
\Promise{$\ket\psi$ has fidelity at least $\cos^2(\pi/8) + \gamma$ with a stabilizer state $\ket{\phi}$}
\KwOutput{$\ket{\phi}$ with probability at least $1 - \delta$}
Let $m \coloneqq \frac{8+4\sqrt{3}}{\cos^8(\pi/8)}\left(n + \log(3/\delta)\right)$.
\label{step:one}

Bell difference sample $m$ times to obtain $x_1,\ldots,x_m \in \F_2^{2n}$.
\label{step:two}

Using the algorithm in \cref{thm:weyl-estimation}, estimate $\braket{\psi|W_{x_i}|\psi}^2$ for each $i$ to accuracy $\pm 2 \sqrt{2}\gamma$ with failure probability at most $\delta/3$.
\label{step:three}

Discard any $x_i$'s for which the estimate of $\braket{\psi|W_{x_i}|\psi}^2$ is less than $\frac{1}{2}$. Let $S$ be the subspace spanned by the remaining samples. If $S$ is not Lagrangian, then output ``\textsc{failure}''. Otherwise, find a Clifford circuit that measures in the stabilizer basis induced by $S$.
\label{step:four}

Measure $4\log(3/\delta)$ copies of $\ket{\psi}$ in the stabilizer basis induced by $S$ and output the majority result.
\label{step:five}
\end{algorithm}

\begin{theorem}\label{thm:bounded-distance}
Let $\ket{\psi}$ be an $n$-qubit state with fidelity at least $\cos^2(\pi/8) + \gamma$ with a stabilizer state $\ket{\phi}$ for $\gamma > 0$.
Given $O\left( n + \frac{\log(n/\delta)}{\gamma^2} \right)$ copies of $\ket\psi$ and $O\left( n^3 + \frac{n^2 \log(n/\delta) + n \log^2(1/\delta)}{\gamma^2}   \right)$ time, \cref{alg:main-alg} outputs $\ket\phi$ with probability at least $1-\delta$.
\end{theorem}
\begin{proof}
Let $m \coloneqq \frac{8+4\sqrt{3}}{\cos^8(\pi/8)}\left(n + \log(3/\delta)\right)$ as in \cref{step:one}.
Since Bell difference sampling requires $4$ copies of $\ket \psi$, the number of copies needed is $4 m = 4\frac{8+4\sqrt{3}}{\cos^8(\pi/8)}\left(n + \log(3/\delta)\right)$ for \cref{lem:fidelity_finds_samples} to guarantee that $\weyl(\ket \phi)$ is generated by some subset of the samples in \cref{step:two} with probability at least $1-\delta/3$.

Our goal now is to filter out the samples outside of $\weyl(\ket \phi)$.
To do so, we estimate the squared expectation of each of the sampled Weyl operators to within error $\eps \coloneqq 2 \sqrt{2}\gamma$, using \cref{thm:weyl-estimation}.
This requires 
\[
\frac{4\log\left(6m / \delta\right)}{\eps^2} = \frac{\log\left(6m / \delta\right)}{2\gamma^2}
\]
copies to succeed with probability $1-\delta/3$.
Conditioned on the success of \cref{step:three}, \Cref{prop:always-larger} tells us that at the end of \cref{step:four}, we will have filtered out all elements outside of $\weyl(\ket \phi)$.
Assuming \cref{step:two,step:three} both succeed, we will be left with $S = \weyl(\ket \phi)$, which is Lagrangian.

Finally, we need to bound how many samples are necessary to measure $\ket \phi$ for the majority result when we measure in the basis specified by $S$.
Let $k$ be the number of times we measure in \cref{step:five}, and for indicator random variables $X_1, \cdots, X_k$ let $X_i = 1$ if and only if measurement $k$ was $\ket \phi$. Let $\mu \coloneqq \Ex\left[\sum_{i=1}^k X_i \right] > \cos^2(\pi/8) \cdot k$. Then via Hoeffding's inequality (\cref{fact:hoeffding}):
\begin{align*}
    \Pr\left[\sum_{i=1}^k X_i \leq \frac{k}{2}\right] < \Pr\left[\sum_{i=1}^k X_i \leq \mu - \cos^2(\pi/8)k  + \frac{k}{2} \right]
    \leq \exp\left(-\frac{k}{4}\right).
\end{align*}
Therefore, by taking $k \geq 4\log(3/\delta)$, we can guarantee that we output the correct state $\ket \phi$ with probability at least $1-\delta/3$.

Via the union bound, the failure probability is at most $\delta$.

In total, we use 
\[
4\frac{8+4\sqrt{3}}{\cos^8(\pi/8)}\left(n + \log(3/\delta)\right) + \frac{\log(6m/\delta)}{2\gamma^2} + 4\log(3/\delta)
\]
copies.

For time complexity, 
performing all of the Bell sampling takes $O(mn) = O(n^2 + n\log(1/\delta))$ time in total. 
Running \cref{thm:weyl-estimation} takes 
\[
O\left(m \cdot n \cdot \frac{\log(m/\delta)}{\gamma^2}\right) = O\left(  \frac{n  \left(n + \log(1/\delta)\right)\log(n/\delta)}{\gamma^2}\right) = O\left(\frac{n^2 \log(n/\delta) + n \log^2(1/\delta)}{\gamma^2}\right) 
\]
time. 
In \cref{step:four}, determining whether $S$ is Lagrangian and computing a Clifford circuit that measures in the basis induced by $S$ can be done using Gaussian elimination on an $m \times 2n$ matrix in $O\left(mn \cdot \min(m, n)\right)$ time \cite[Section VI]{aaronson2004simulation}.\footnote{For further detail, see also \cite[Section 3]{grewal2023efficient}.} Furthermore, the computed Clifford circuit contains at most $O(n^2)$ gates. Thus, \cref{step:four} takes
\[
O\left( n^3 + n^2 \log(1/\delta) \right)
\]
time.
Finally, measuring in the Clifford basis takes $O(n^2)$ time, so \cref{step:five} takes $O(n^2 \log(1/\delta))$ time.
Overall, the time complexity is
\[
O\left( n^3 + \frac{n^2 \log(n/\delta) + n \log^2(1/\delta)}{\gamma^2}   \right).\qedhere
\]
\end{proof}

\section{Tolerant Property Testing of Stabilizer States}

We collect (and give alternative proofs of) a few results related to the property testing algorithm for stabilizer states due to Gross, Nezami, and Walter \cite{gross2021schur} (hereafter, the ``GNW algorithm''). 
We combine these results with the prior work of \cite{grewal2023efficient} to
give a \emph{tolerant property testing} algorithm for stabilizer states. 
By tolerant property testing, we mean that the tester must accept inputs that are $\eps_1$-close to having some property and reject inputs that are $\eps_2$-far from having the same property. 
This is more general than the standard setting where $\eps_1$ is set to $0$.  

The following two remarks are important for understanding the extent of our contribution. First, our algorithm is similar to the algorithm given in \cite{grewal_et_al:LIPIcs.ITCS.2023.64}\footnote{Which is itself a repeated application of the base GNW algorithm, for the purposes of error amplification.} that distinguishes Haar-random states from quantum states with at least $1/\poly(n)$ stabilizer fidelity. Second, we note that our algorithm only works in certain parameter regimes, not for all sensible settings of $\eps_1$ and $\eps_2$. This is discussed further in \cref{subsec:parameter-regime}.

To explain our property testing model in more detail, we are testing whether or not a quantum state is close to a stabilizer state, where distance is measured with fidelity. 
Specifically, we are given copies of an $n$-qubit quantum pure state $\ket\psi$ as input, and we must decide whether $F_\stabset(\ket\psi) \geq 1- \eps_1$ or $F_\stabset(\ket\psi) \leq 1 - \eps_2$, promised that one of them is the case.   

\subsection{Completeness and Soundness Analysis}
We record several statements that will go into our proof that our tolerant property testing algorithm is sound. 
It is first convenient to recall the GNW algorithm, which works as follows. 
Perform Bell difference sampling on the input state to get a string $x \in \F_2^{2n}$. Then perform the $\{\pm 1\}$ measurement $W_x^{\otimes 2}$ on $\ket{\psi}^{\otimes 2}$ and accept if the result is $1$.
The algorithm uses six copies of the input state.

Following the notation of \cite{grewal_et_al:LIPIcs.ITCS.2023.64} let: 
\begin{align}\label{eq:eta-definition}
\eta \coloneqq \E_{x \sim q_\psi}[2^n p_\psi(x)],
\end{align}
be the expected value of such a measurement.
We use the following identity proven by \cite{gross2021schur} and \cite{grewal_et_al:LIPIcs.ITCS.2023.64}.

\begin{fact}[{\cite[Fact 16]{grewal_et_al:LIPIcs.ITCS.2023.64}, \cite[Section 3.1]{gross2021schur}}]\label{fact:eta-p-cubed}
    Let $\ket \psi$ be an $n$-qubit quantum pure state. Then,
    \[\eta = 4^n \sum_{x \in \F_2^{2n}} p_\psi(x)^3,\]
    where $\eta$ is the expected value of the GNW algorithm measurement defined in \cref{eq:eta-definition}.
\end{fact}
\begin{proof}
    \begin{align*}
        \eta \coloneqq \E_{x \sim q_\psi}[2^n p_\psi(x)]
        &= 2^n \sum_{x \in \F_2^{2n}} q_\psi(x) p_\psi(x)\\
        &= 8^n \sum_{x \in \F_2^{2n}} \widehat{q}_\psi(x) \widehat{p}_\psi(x) && (\text{\nameref{fact:plancherel}})\\
        &= 32^n \sum_{x \in \F_2^{2n}} \widehat{p}_\psi(x)^3 && (\text{\cref{thm:convolution-theorem}})\\
        &= 4^n \sum_{x \in \F_2^{2n}} p_\psi(x)^3. && (\text{\cref{fact:weyl-fourier-duality}})\qedhere
    \end{align*}
\end{proof}
Prior to this work, it was known that, for any quantum pure state $\ket\psi$, 
\[
(4 \eta - 1)/3 \leq F_\stabset(\ket{\psi}) \leq \eta^{1/6}, 
\]
where the lower bound is due to \cite{gross2021schur} and the upper bound is due to \cite[Lemma 15]{grewal_et_al:LIPIcs.ITCS.2023.64}.\footnote{The lower bound is discussed only at the end of \cite[Section 3.1]{gross2021schur}. } In this subsection, we re-prove these bounds 
using the formalism and techniques developed in this work. 
We begin by giving an alternative proof of the upper bound. 

\begin{lemma}
\label{lemma:improved-completeness}
Let $\ket\psi$ be an $n$-qubit quantum state. 
\[
F_\calS(\ket\psi) \leq \eta^{1/6}.
\]
\end{lemma}
\begin{proof}
Let $\ket\phi$ be the stabilizer state that achieves the maximum overlap with $\ket\psi$ and let $S^* \coloneqq \weyl(\ket\phi)$.
Then 
\begin{align*}
F_\stabset(\ket{\psi})\leq \sqrt{\sum_{x \in S^*} p_\psi(x)}  \leq \left(4^n \sum_{x \in S^*} p_\psi(x)^3\right)^{1/6} \leq \left(\sum_{x \in \F_2^{2n}} p_\psi(x)^3\right)^{1/6} = \eta^{1/6}.
\end{align*}
The first inequality follows from \cref{lem:f_lowerbound}, the second inequality follows from H\"older's inequality, the third inequality follows from the nonnegativity of $p_\psi$, and the final equality follows from \cref{fact:eta-p-cubed}.
\end{proof}

We now move on to the lower bound which says that $(4 \eta - 1)/3 \leq F_\stabset(\ket{\psi})$.
We begin by proving a lower bound on stabilizer fidelity in terms of $p_\psi$-mass.  

\begin{lemma}[{\cite[Lemma 4.6]{grewal2023efficient}}]\label{lem:product-state-approximation-general}
Let $T$ be an isotropic subspace of dimension $n-t$, and suppose that 
    \[\sum_{x \in T} p_\psi(x) \geq \frac{1 - \eps}{2^t}.\]  
    Then there exists a state $\ket{\hat{\psi}}$ with $T \subseteq \weyl(\ket{\hat{\psi}})$ such that the fidelity between $\ket{\hat{\psi}}$ and $\ket{\psi}$ is at least $1 - \eps$. 
\end{lemma}

\begin{corollary}\label{cor:product-state-approximation-general}
    For any $n$-qubit quantum state $\ket\psi$ and Lagrangian subspace $T$, 
    \[
    F_\calS(\ket\psi) \geq \sum_{x \in T} p_\psi(x).
    \]
\end{corollary}
\begin{proof}
We apply \cref{lem:product-state-approximation-general} with the Lagrangian subspace $T$ (so that $t = 0$).
Following the notation in \cref{lem:product-state-approximation-general}, 
there must exist a state $\ket{\hat{\psi}}$ whose fidelity with $\ket\psi$ is at least $\sum_{x \in T} p_\psi(x)$. 
Note also that $\ket{\hat{\psi}}$ is a stabilizer state since its unsigned stabilizer group is Lagrangian.
We conclude that the stabilizer fidelity $F_\calS(\ket\psi)$ must be at least $\sum_{x \in T} p_\psi(x)$ because the stabilizer fidelity is the maximum fidelity over all stabilizer states.
\end{proof}

We also need the following fact: 
any pair of Weyl operators whose expectation with $\ket\psi$ is each at least $1/2$ must commute. 

\begin{fact}[{\cite{gross2021schur,grewal2023efficient}}]\label{fact:M_half_commute}
    Let $M = \{x \in \F_2^{2n} : 2^n p_\psi(x) > \frac{1}{2}\}$. Then for all $x, y \in M$, $[x, y] = 0$.
\end{fact}

We now prove the lower-bound.

\begin{lemma}\label{prop:improved-soundness}
Let $\ket\psi$ be an $n$-qubit  pure state. Then
\[
\frac{4 \eta - 1}{3} \leq F_\stabset(\ket{\psi}). 
\]
\end{lemma}

\begin{proof}
Let $M \coloneqq \{ x \in \F_2^{2n} : 2^n p_\psi(x) > 1/2 \}$. 
By \cref{fact:M_half_commute}, $M$ is isotropic. 
We can arbitrarily complete $M$ to some Lagrangian subspace $M^\prime \supseteq M$, and, since $p_\psi$ only takes non-negative values, it is obvious that $\sum_{x \in M^\prime} p_\psi(x) \geq \sum_{x \in M} p_\psi(x)$. 
Furthermore, by \cref{cor:product-state-approximation-general}, we know that \[F_\stabset(\ket\psi) \geq \sum_{x \in M^\prime} p_\psi(x) \geq \sum_{x \in M} p_\psi(x).\]
All that remains is proving that 
$\sum_{x \in M} p_\psi(x) \geq (4\eta - 1)/3$.
\begin{align*}
\sum_{x \in M}p_\psi(x) 
&= \Pr_{x \sim p_\psi}\left[x \in M\right]\\
&= \Pr_{x \sim p_\psi}\left[2^n p_\psi(x) > 1/2\right] && \text{(Definition of $M$)}\\
&= \Pr_{x \sim p_\psi}[4^{n}p_\psi^{2}(x) > 1/4]\\
&= 1 - \Pr_{x \sim p_\psi}[4^{n}p_\psi^{2}(x) \leq 1/4]\\
&= 1-\Pr_{x \sim p_\psi}[1-4^{n}p_\psi^{2}(x) \geq 3/4]\\
&\geq 1- \frac{4}{3}\left(1 - \E_{x \sim p_\psi}[4^n p_\psi^2(x)] \right) && \text{(Markov's Inequality)}\\
&= 1- \frac{4}{3}\left(1 - \eta \right) && (\text{\cref{fact:eta-p-cubed}})\\
&= \frac{4 \eta - 1}{3}.&&\qedhere 
\end{align*}
\end{proof}

We note that \cite{gross2021schur} gave an alternative proof of the fact that $F_\stabset(\ket\psi) \geq \sum_{x \in M} p_\psi(x)$ for $M \coloneqq \{ x \in \F_2^{2n} : 2^n p_\psi(x) > \frac{1}{2} \}$.
Our proof uses the more general \cref{cor:product-state-approximation-general}.

\subsection{The Algorithm}\label{subsec:tolerant-testing}

In the previous subsection, we established that for all quantum states $\ket{\psi}$,
\[
\frac{4 \eta - 1}{3} \leq F_\stabset(\ket{\psi}) \leq \eta^{1/6}.
\]
To simplify notation, let $\alpha_1 \coloneqq 1 - \eps_1$ and $\alpha_2 \coloneqq 1 - \eps_2$.
Observe that if $F_\stabset(\ket{\psi}) \geq \alpha_1$ then $\eta \geq \alpha_1^6$, and if $F_\stabset(\ket{\psi}) \leq \alpha_2$ then $\eta \leq \frac{3 \alpha_2 + 1}{4}$. 
This is the basis of our testing algorithm.
Specifically, as long as 
\[
\alpha_1^6 - \frac{3 \alpha_2+ 1}{4} \geq \frac{1}{\poly(n)},
\]
then we can efficiently distinguish the two cases simply by estimating $\eta$. 
For the remainder of this section, define 
\begin{align*}
\gamma \coloneqq \alpha_1^6 - \frac{3 \alpha_2+ 1}{4}. 
\end{align*}
Our algorithm is stated in \Cref{alg:tolerant-tester}.

\vspace{\baselineskip}
\begin{algorithm}[H]
\SetKwInOut{Promise}{Promise}
\caption{Tolerant Property Testing of Stabilizer States}
\label{alg:tolerant-tester}
\DontPrintSemicolon
\KwInput{$48\log(2/\delta)/\gamma^2$ copies of $\ket\psi$}
\Promise{Either case (i): $F_\stabset(\ket{\psi})\geq \alpha_1$ or case (ii): $F_\stabset(\ket\psi) \leq \alpha_2$, for $\alpha_1, \alpha_2 \in [0,1]$ such that $\gamma > 0$}
\KwOutput{$1$ if case (i) holds and $0$ if case (ii) holds, with probability at least $1-\delta$}
Let $m =  \frac{8\log(2/\delta)}{\gamma^2}$.

\RepTimes{$m$}{
    Perform Bell difference sampling to obtain $W_x \sim q_{\psi}$.

    Perform the measurement $W_x^{\otimes 2}$ on $\ket{\psi}^{\otimes 2}$. Let $X_i \in \{\pm 1\}$ denote the measurement outcome.
}

Set $\hat{\eta} = \frac{1}{m} \sum_i X_i$. 
Output $1$ if $\hat{\eta} > \alpha_1^6 - \frac{\gamma}{2}$ and $0$ otherwise. 
\end{algorithm}
\vspace{\baselineskip}

\begin{theorem}\label{thm:tolerant-tester}
For $\gamma > 0$, \Cref{alg:tolerant-tester} is correct and it uses $48\log(2/\delta)/\gamma^2$ copies of the input state, $O(n \log(1/\delta)/\gamma^2)$ time, and succeeds with probability at least $1 - \delta$.
\end{theorem}
\begin{proof}
\Cref{alg:tolerant-tester} fails when $|\hat\eta - \eta| \geq \gamma/2$. 
    \cite{grewal_et_al:LIPIcs.ITCS.2023.64} proved that $\hat{\eta} = \frac{1}{m} \sum_i X_i$ is an unbiased estimator of $\eta$ (i.e., $\Ex[\hat{\eta}] = \eta$). 
Therefore, by Hoeffding's inequality (\cref{fact:hoeffding}),
    \[
    \Pr[\text{\Cref{alg:tolerant-tester} fails}] = 
    \Pr[|\hat \eta - \eta| \geq \gamma/2] \leq 2e^{-m\gamma^2/8} = \delta.
    \]
    
The number of copies follows from the fact that Bell difference sampling consumes $4$ copies of the input state, the measurement in Step 4 of \Cref{alg:tolerant-tester} consumes $2$ copies of the input state, and that the loop is repeated $m$ times. The running time is clearly $O(mn)$.
\end{proof}

\subsection{Parameter Regime Discussion}\label{subsec:parameter-regime}

\begin{figure}[t!]
    \centering
    \begin{subfigure}{1.0\textwidth}
    \centering 
    \includegraphics[width=.55\linewidth]{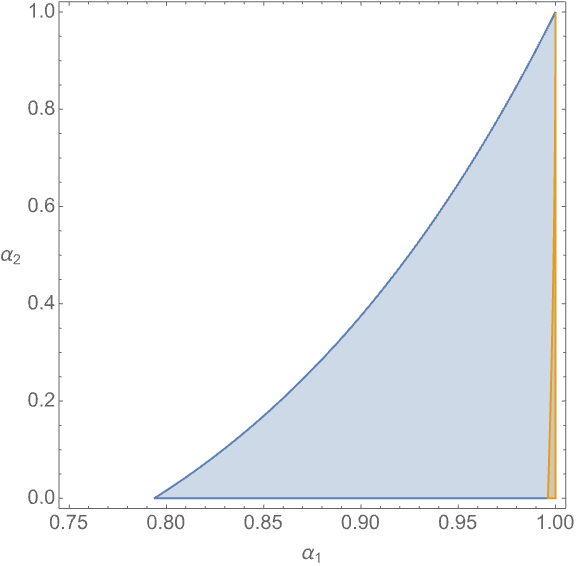}
    \hfill
    \end{subfigure}
    \caption{
    The shaded regions indicate the parameter regimes of $\alpha_1$ and $\alpha_2$ that are permissible by the analysis of the GNW algorithm in \Cref{subsec:parameter-regime} (orange) and \cref{alg:tolerant-tester} (blue). 
    Thus, the difference between the orange and blue regions illustrates the improvement due to \cref{lemma:improved-completeness}.}
    \label{fig:parameter-regime}
\end{figure}

We conclude this section by studying the regime in which our tolerant testing algorithm works in comparison to prior work.
The GNW algorithm already implicitly functions as a tolerant property tester: because it uses $6$ copies of $\ket{\psi}$ and accepts any stabilizer state with probability $1$, if the trace distance between $\ket{\psi}$ and some stabilizer state is at most $\eps$, then the test accepts $\ket{\psi}$ with probability at least $1 - 6\eps$. We can use this observation to establish the values of $\alpha_1$ and $\alpha_2$ in which repeated applications of the GNW algorithm works, given only the soundness analysis of \cite{gross2021schur} but not applying the completeness bound implied by \Cref{lemma:improved-completeness}.

Let $p_{\mathsf{accept}}$ denote the acceptance probability of the GNW algorithm.
It is easy to show that $\eta = 2 p_{\mathsf{accept}} -1$ (see \cite[Page 19]{gross2021schur}).
As mentioned above, \cite{gross2021schur} proved that for any quantum state $\ket{\psi}$, $\frac{4 \eta - 1}{3} \leq F_\stabset(\ket\psi)$. 
Additionally, since the GNW algorithm uses 6 copies of the input state and accepts stabilizer states with probability $1$, it follows that $1 - 6\sqrt{1 - F_\stabset(\ket{\psi})} \leq p_{\mathsf{accept}}$, where we are using the fact that the trace distance between $\ket\psi$ and the stabilizer state maximizing fidelity is $\sqrt{1 - F_\stabset(\ket\psi)}$. 
Finally, using the fact that $\eta = 2 p_{\mathsf{accept}} -1$, we get
$F_\stabset(\ket\psi) \leq \frac{1}{144}(2\eta - \eta^2 + 143)$.

In a ``yes'' instance, where we are promised that $F_\calS(\ket\psi) \geq \alpha_1$, we 
have the inequality $\frac{1}{144}(2\eta - \eta^2 + 143) \geq \alpha_1$.
Solving for $\eta$ gives $\eta \geq 1-12\sqrt{1-\alpha_1}$.
Similarly, in a ``no'' instance, where we are promised that $F_\calS(\ket\psi)\leq \alpha_2$, we have $\eta \leq (3\alpha_2 +1)/4$.
Hence, following the same argument as in \cref{subsec:tolerant-testing}, the GNW algorithm tolerantly tests stabilizer states as long as:
\[
1 - 12\sqrt{1-\alpha_1} >\frac{3\alpha_2 + 1}{4}, 
\]
whereas, as shown earlier, our algorithm works as long as 
\[
\alpha_1^6 > \frac{3 \alpha_2+ 1}{4}.
\]
This is a significant improvement, which is shown visually in \cref{fig:parameter-regime}.

\section{Discussion and Open Problems}

A natural direction for future work is to improve the performance of our algorithms or to prove (conditional or unconditional) lower bounds. 
In particular, can the exponential running time of \cref{alg:fidelity-estimation} be improved upon, or is stabilizer state approximation computationally hard for general parameter regimes? We are optimistic that the exponential factors in our runtime analysis could be made much smaller in practice, because our bound on the sample complexity of finding a complete set of generators is probably far from optimal.

We also remark that, at least superficially, our problem of finding the nearest stabilizer state resembles the closest vector problem (CVP): given a lattice $L$ and a target vector, find the nearest lattice point to the target vector. In our problem, we are given a target vector, and we want to find the nearest stabilizer state to the target vector. While not a lattice, the stabilizer states are ``evenly spread'' across the complex unit sphere due to their 3-design property \cite{kueng2015qubit, webb2016clifford,zhu2017multiqubit}.
CVP is known to be $\NP$-hard to solve approximately to within any constant and some almost-polynomial factors \cite{van1981another, arora1997hardness, dinur1998approximating}. Is there a formal connection between these two problems?

Can tighter bounds between $\eta$ and stabilizer fidelity be proven? In \cite[Appendix B]{grewal_et_al:LIPIcs.ITCS.2023.64}, the authors prove that one can hope for at most a roughly quadratic improvement in the bound $\fidelity_\stabset(\ket\psi)^6 \leq \eta$. 
In addition to $\eta$, are there other statistics related to stabilizer fidelity (or any other stabilizer complexity measure) that can be estimated efficiently?
Progress in this direction would extend the parameter regimes for which our property testing algorithm works (see \cref{fig:parameter-regime}).

One can view the output of \cref{alg:fidelity-estimation} as an approximation of the input state by a nearby stabilizer state. 
Following this theme, a natural objective is to design similar approximation algorithms relative to other classes of quantum states such as product states or matchgate states. 
We note that it is even open to design a time-efficient algorithm that, given copies of an $n$-qubit quantum state, outputs the nearest state from the set $\{\ket{0}, \ket{1}, \ket{+}, \ket{-}, \ket{i}, \ket{-i}\}^{\otimes n}$, which is a subset of stabilizer states. 
In addition to potentially improving Clifford+$T$ simulation algorithms (as discussed in \cref{subsec:our-results}), are there other applications for these types of state approximation algorithms?

\section*{Acknowledgments}
We thank David Gosset, Srinivasan Arunachalam, Sepehr Nezami, and Arkopal Dutt for helpful conversations.
SG, VI, DL are supported via Scott Aaronson by a Vannevar Bush Fellowship from the US Department of Defense, the Berkeley NSF-QLCI CIQC Center, a Simons Investigator Award, and the Simons ``It from Qubit'' collaboration. WK is supported by an NDSEG Fellowship, and also acknowledges support from the U.S. Department of Energy, Office of Science,
National Quantum Information Science Research Centers, Quantum Systems Accelerator.
VI is supported by an NSF Graduate Research Fellowship.
DL is also supported by NSF award FET-2243659.
\bibliographystyle{alphaurl}
\bibliography{refs}

\appendix

\section{Estimating Squared Expectation of Weyl Operators}\label{appendix:hkp}

We provide an alternative proof of \cref{thm:weyl-estimation}, which was first proved in \cite{huang2021information}. 
The remarkable simplicity of our proof illustrates the power of the duality theorems in \cref{sec:duality} and underscores the usefulness of our Fourier-analytic techniques. 
For convenience, we restate the theorem.

\hkpthmrestate*

Let $t_\psi(x)$ be the probability of Bell sampling (without the difference) $x \in \F_2^{2n}$. The following was shown by Montanaro.

\begin{proposition}[{\cite[Lemma 2]{montanaro-bell-sampling}}]\label{prop:montanaro-bell-sampling}
Let $\ket \psi$ be an $n$-qubit quantum state. Bell sampling on two copies of $\ket \psi$ outputs $x \in \F_2^{2n}$ with probability
\[
t_\psi(x) \coloneqq \frac{1}{2^n}\abs{\braket{\psi | W_x | \psi^*}}^2
\]
where $\ket{\psi^*}$ is the state produced by taking the complex conjugate of each amplitude in the computational basis of $\ket \psi$.
\end{proposition}

Among the single-qubit Pauli matrices, only $Y$ has complex entries; $I$, $X$, and $Z$ are invariant under complex conjugation. $Y$ also has the property that its complex conjugate is its negation. For this reason, switching between $\ketbra{\psi}{\psi}$ and $\ketbra{\psi^*}{\psi^*}$ has the effect of flipping the sign of the coefficients $c_\psi$ in the Weyl expansion (\cref{def:weyl-expansion}) each time a Pauli-$Y$ appears.
We formalize this below, starting with a function that ``counts'' how many times $Y$ appears.

\begin{definition}
For $x = (a,b) \in \F_2^{2n}$, define $\pi(x) \coloneqq a \cdot b $.
\end{definition}

From this, we can see that $c_{\psi^*}(x) = c_\psi(x) (-1)^{\pi(x)}$. This allows us to prove the following about the Fourier decomposition of $t_\psi$.

\begin{corollary}[{\cite[Proposition 3.15]{damanik2018optimality}}]\label{cor:damanik}
For any $n$-qubit quantum state $\ket\psi$,
\[t_\psi(x) = \frac{1}{2^n}\sum_{y \in \F_2^{2n}} p_\psi(y) (-1)^{\pi(y) + [x, y]}.\]
\end{corollary}
\begin{proof}
\begin{align*}
t_\psi(x) &= \frac{1}{2^n} \abs{\braket{\psi | W_x | \psi^*}}^2  && \text{(\cref{prop:montanaro-bell-sampling})}\\
&= \frac{1}{2^n} \tr\left[W_y \ketbra{\psi^*}{\psi^*} W_x \ketbra{\psi}{\psi}\right] && (\text{Cyclic property of trace})\\
&= \frac{1}{4^n} \tr\left[W_y \left(\sum_{y \in \F_2^{2n}} (-1)^{\pi(y)} c_\psi(y) W_y\right) W_x \left( \sum_{z \in \F_2^{2n}} c_\psi(z) W_z\right) \right] && (\text{Weyl expansion})\\
&= \frac{1}{4^n} \sum_{y, z \in \F_2^{2n}} (-1)^{\pi(y)} c_\psi(y) c_\psi(z) \tr\left[W_x W_y W_x W_z \right]\\
&= \frac{1}{4^n} \sum_{y, z \in \F_2^{2n}} (-1)^{\pi(y) + [x, y]} c_\psi(y) c_\psi(z) \tr\left[W_y W_z \right] && (W_x W_y = (-1)^{[x, y]}W_y W_x)\\
&= \frac{1}{2^n} \sum_{y \in \F_2^{2n}} c_\psi(y)^2 (-1)^{\pi(y) + [x, y]} && (\tr\left[W_x W_y \right] = 2^n \cdot \indic{x = y})\\
&= \frac{1}{2^n} \sum_{y \in \F_2^{2n}} p_\psi(y) (-1)^{\pi(y) + [x, y]}. && (p_\psi(x) \coloneqq c_\psi(x)^2) && \qedhere
\end{align*}
\end{proof}

We can now prove a duality theorem for $t_\psi$, just as we did for $p_\psi$ and $q_\psi$ in \cref{thm:p-mass-identity-subgroups,thm:q-mass-identity-subgroups}, respectively.

\begin{theorem}\label{thm:bell-diff-duality}
    For a subspace $S \subseteq \F_2^{2n}$:
    \[
        \sum_{x \in S} t_\psi(x) = \frac{\abs{S}}{2^n} \sum_{x \in S^\sympcomp} p_\psi(x)(-1)^{\pi(x)}.
    \]
\end{theorem}
\begin{proof}
    \begin{align*}
        \sum_{x \in S} t_\psi(x) &= \sum_{x \in S} \frac{1}{2^n}\sum_{y \in \F_2^{2n}} p_\psi(x) (-1)^{\pi(y) + [x, y]} && (\text{\cref{cor:damanik}})\\
        &= \frac{1}{2^n}\sum_{y \in \F_2^{2n}} p_\psi(x) (-1)^{\pi(y)} \sum_{x \in S}(-1)^{[x, y]}\\
        &= \frac{\abs{S}}{2^n} \sum_{x \in S^\sympcomp} p_\psi(x)(-1)^{\pi(x)}. && (\text{\cref{lemma:sum-over-characters}}) \qedhere
    \end{align*}
\end{proof}

We are now ready to prove \cref{thm:weyl-estimation}.

\begin{proof}[Proof of \cref{thm:weyl-estimation}]
Because $\braket{\psi | I |\psi}^2 = 1$ always, we can assume without loss of generality that all of the Weyl operators $P_1,\ldots,P_m$ are non-identity.

Let $S_x \coloneqq \langle x \rangle$ be the one-dimensional subspace spanned by $x \in \F_2^{2n} \setminus \{0^{2n} \}$.
Define $\eta_x$ as the probability that Bell sampling chooses a $y$ for which $W_x$ and $W_y$ commute, i.e.,
\[\eta_x \coloneqq \Pr_{y \sim t_\psi}\big[[x, y] = 0\big] = \sum_{y \in S_x^\sympcomp} t_\psi(y) = \frac{\abs{S_x^\sympcomp}}{2^n}\sum_{y \in S_x} p_\psi(y) (-1)^{\pi(y)} = \frac{1 + (-1)^{\pi(x)}\braket{\psi | W_x | \psi}^2}{2},\]
where the second equality follows from \cref{thm:bell-diff-duality}.
Thus we can estimate $\braket{\psi | W_x | \psi}^2$ via the relationship:
\[\braket{\psi | W_x | \psi}^2 = (-1)^{\pi(x)}\left(2\eta_x - 1\right).\]

By linearity of expectation, we only need to estimate $\eta_x$, a Bernoulli random variable, to accuracy $\pm\eps/2$ for our estimate of $\braket{\psi | W_x | \psi}^2$ to be within $\pm\eps$. 
By Hoeffding's inequality (\cref{fact:hoeffding}), we can estimate a Bernoulli random variable to accuracy $\pm \eps/2$ using only $\frac{2\log(2/\delta)}{\eps^2}$ samples, regardless of the value of $\braket{\psi | W_x | \psi}^2$. 
To estimate the squared expectation of all $m$ Weyl operators to within $\eps$ error, we can simply use $\frac{2\log(2m/\delta)}{\eps^2}$ samples and appeal to the union bound. 
Since Bell sampling uses $2$ copies of $\ket \psi$, we use $\frac{4\log(2m/\delta)}{\eps^2}$ copies of $\ket\psi$ in total.
Finally, because performing each Bell measurement, computing the symplectic product, and computing $\pi(x)$ all require $O(n)$ time, estimating any particular squared expectation value takes $O\left(\frac{4n\log(2m/\delta)}{\eps^2}\right)$ time. Hence, the algorithm uses a grand total of $O\left(\frac{4mn\log(2m/\delta)}{\eps^2}\right)$ time.
\end{proof}

\end{document}